\documentclass[11pt]{article}
\newif\ifdraft
\drafttrue
\newcommand{\ZGnote}[1]{\ifdraft\textcolor{blue}{[Zeyu: #1]}\fi}

\newcommand{\ZSnote}[1]{\ifdraft\textcolor{green}{[Chong: #1]}\fi}

\usepackage{geometry}
\usepackage{xcolor}
\usepackage{amsfonts}
\usepackage{amsmath}
\usepackage{amssymb}
\usepackage{keytheorems}
\usepackage{cellspace} 
\definecolor{myRed}{RGB}{187,0,0}
\usepackage[colorlinks = true, citecolor=myRed,
urlcolor=myRed,
linkcolor=myRed,
]{hyperref}
\usepackage[capitalize]{cleveref}

\newkeytheorem{theorem}[name=Theorem,parent=section]
\newkeytheorem{lemma}[name=Lemma,sibling=theorem]
\newkeytheorem{proposition}[name=Proposition,sibling=theorem]
\newkeytheorem{corollary}[name=Corollary,sibling=theorem]
\newkeytheorem{claim}[name=Claim,sibling=theorem]
\newkeytheorem{fact}[name=Fact,sibling=theorem]
\newkeytheorem{question}[name=Question]

\newkeytheorem{remark}[name=Remark,sibling=theorem]

\newkeytheorem{definition}[name=Definition,sibling=theorem,style=definition]

\newcommand{\eps}{\varepsilon}
\renewcommand{\epsilon}{\eps}

\newcommand{\ex}{\mathop{{}\mathbb{E}}}
\newcommand{\rank}{\mathrm{rank}}

\newcommand{\F}{\mathbb{F}}
\newcommand{\C}{\mathbb{C}}
\newcommand{\Z}{\mathbb{Z}}

\newcommand{\qb}[2]{\genfrac{(}{)}{0pt}{}{#1}{#2}_q}
\newcommand{\Span}{\mathrm{span}}
\newcommand{\poly}{\mathrm{poly}}
\newcommand{\supp}{\mathrm{supp}}
\newcommand{\GL}{\mathrm{GL}}
\newcommand{\Sym}{\mathrm{Sym}}
\newcommand{\Tr}{\mathrm{Tr}}
\newcommand{\VBP}{\mathrm{VBP}}
\newcommand{\VP}[3]{\VBP_{#1,#2,#3}^1}

\DeclareMathOperator{\rowspan}{rowspan}
\DeclareMathOperator{\colspan}{colspan}
\DeclareMathOperator{\PG}{PG}

\begin{document}
\date{}
\title{Explicit Rank Extractors and Subspace Designs via Function Fields, with Applications to Strong Blocking Sets}

\author{
Zeyu Guo\footnotemark[1]
\and
Roshan Raj\footnotemark[1]
\and
Chong Shangguan\footnotemark[2]
\and
Zihan Zhang\footnotemark[1]
}

\footnotetext[1]{The Ohio State University, \texttt{zguotcs@gmail.com}, \texttt{amritanshus128@gmail.com}, \texttt{zhang.13691@osu.edu}}
\footnotetext[2]{Shandong University, \texttt{theoreming@163.com}}
\makeatletter
\def\thanks#1{\protected@xdef\@thanks{\@thanks
        \protect\footnotetext{#1}}}
\makeatother


\maketitle
\begin{abstract}
We give new explicit constructions of several fundamental objects in linear-algebraic pseudorandomness and combinatorics, including lossless rank extractors, subspace designs, and strong $s$-blocking sets over finite fields.

Our focus is on the small-field regime, where the field size depends only on a secondary parameter (such as the rank or codimension) and is independent of the ambient dimension. This regime is central to several applications, yet remains poorly understood from the perspective of explicit constructions.

In this setting, we obtain the first explicit constructions of lossless rank extractors and subspace designs when the target rank or codimension $r$ is much smaller than the ambient dimension $k$, over finite fields $\mathbb{F}_q$ with $q \ge \mathrm{poly}(r)$ and $q$ non-prime, with near-optimal parameters. For other finite fields, including prime fields and small fields, we obtain weaker but still improved bounds.

As a consequence, we construct explicit strong $s$-blocking sets in $\mathrm{PG}(k-1,q)$ of size $O(s(k-s)q^s)$ for all sufficiently large non-prime fields $q \ge \mathrm{poly}(s)$, matching the best known non-explicit bounds up to constant factors. This significantly improves the previous best bound $2^{O(s^2 \log s)} q^s k$ of Bishnoi and Tomon (Combinatorica, 2026), which requires $q \ge 2^{\Omega(s)}$.

Our main algebraic constructions are function-field analogs of the rank-extractor constructions of Gabizon--Raz and Forbes--Shpilka. We further combine these constructions with PIT-based field reduction. In addition, we develop a complementary Fourier-analytic framework based on $\varepsilon$-biased sets, which yields improved explicit constructions of strong $s$-blocking sets over small fields.
\end{abstract}

\section{Introduction}

A central theme in theoretical computer science is the explicit construction of combinatorial and algebraic objects that match the guarantees of the probabilistic method. In recent years, a rich theory of \emph{linear-algebraic pseudorandomness} has emerged, focusing on structured collections of linear maps that behave like random ones with respect to rank and dimension. Prominent examples include lossless rank extractors, subspace designs, and dimension expanders, which play a key role in polynomial identity testing, coding theory, and derandomization.

A recurring challenge in this area is to obtain explicit constructions over \emph{small finite fields}. While the probabilistic method often yields near-optimal parameters over any field, known explicit constructions typically require the field size to grow with the ambient dimension. Bridging this gap, especially in the regime where the field size depends only on a small parameter, is a central open problem with numerous applications.

In this work, we study this problem for several natural objects in linear-algebraic pseudorandomness and combinatorics, including lossless rank extractors, subspace designs, and strong $s$-blocking sets. Our main algebraic approach uses function fields and polynomial identity testing, and for strong $s$-blocking sets we also develop a complementary Fourier-analytic approach via $\eps$-biased sets.

\paragraph{Lossless rank condensers and extractors.}

A central object in the theory of ``linear-algebraic pseudorandomness'' \cite{FG15} is the notion of (seeded) lossless rank condensers and extractors \cite{GR08,FS12,FSS14,FG15}.

\begin{definition}[Lossless rank condensers and extractors]\label{def:rank-condenser}
Let $r \le t \le k$ be positive integers. A finite collection $\mathcal{E}=\{E_i\}_{i=1}^n \subseteq \F^{t\times k}$ of matrices over a field $\F$ is called a $(k,r,t,L)$ weak lossless rank condenser over $\F$ if for every matrix $M \in \F^{k \times r}$ of rank $r$, the number of indices $i\in [n]$ for which $\rank(E_i M) < r$ is at most $L$.
It is called a $(k,r,t,L)$ strong lossless rank condenser over $\F$ if 
\[
\sum_{i=1}^n (r-\rank(E_i M))\leq L.
\] 
Note that a strong lossless rank condenser is also a weak lossless rank condenser with the same parameters.

We call $n$ the size of $\mathcal{E}$.
When $r=t$, a $(k,r,t,L)$ weak (resp. strong) lossless rank condenser is also called a $(k,r,L)$ weak (resp. strong) lossless rank extractor.
\end{definition}


Lossless rank condensers and extractors were first studied in the weak sense. Gabizon and Raz \cite{GR08} constructed explicit lossless rank extractors, corresponding to the special case $r=t$ in our definition, and used them to obtain explicit affine extractors over large fields. In particular, they gave explicit $(k,r,kr^2)$ lossless rank extractors over finite fields $\F_q$ of any size $n\leq q$. Forbes and Shpilka \cite{FS12} gave a different construction over fields $\F$ with $|\F|>k$ using Vandermonde matrices, achieving $L=kr-\binom{r+1}{2}$. Forbes, Saptharishi, and Shpilka \cite{FSS14} further refined the analysis of this construction and achieved $L=r(k-r)$, which is optimal over algebraically closed fields (see \cite{FZ16,Guo24,brakensiek2025random}).

Guruswami and Forbes \cite{FG15} later introduced the distinction between weak and strong lossless rank condensers and established their correspondence with weak and strong subspace designs, respectively. In particular, Wronskian-based and folded-Wronskian-based constructions, including those arising from earlier work of Forbes and Shpilka \cite{FS12,FSS14} and from constructions of subspace designs by Guruswami and Kopparty \cite{guruswami2016explicit}, yield explicit strong lossless rank condensers through this correspondence.

Beyond their role in constructing affine extractors \cite{GR08}, (weak) lossless rank condensers and extractors play a central role in polynomial identity testing for various models, including depth-3 arithmetic circuits \cite{KS11,KS09,SS13,SS12} and read-once oblivious arithmetic branching programs \cite{FS13,FSS14,AGKS15}. For further applications and a comprehensive discussion, see \cite{FG15}.

Nonlinear variants of these objects have also been studied. Dvir, Gabizon, and Wigderson \cite{DGW09} introduced deterministic rank extractors for polynomial sources and gave explicit constructions. Guo, Volk, Jalan, and Zuckerman \cite{GVJZ23} extended this line of work by defining deterministic rank condensers and extractors for varieties as dimension-preserving maps, yielding explicit constructions for algebraic sources that generalize \cite{DGW09,Dvi12}. They also observed that the objects in Definition \ref{def:rank-condenser} can be viewed as linear seeded rank condensers and extractors for varieties. Finally, Guo \cite{Guo24} introduced variety evasive subspace families as a nonlinear generalization, with further applications to polynomial identity testing.

In this paper, we focus on explicit\footnote{Here and throughout, explicitness means that the family can be deterministically constructed in time polynomial in $q$ and the relevant dimension and output-size parameters. In particular, a generator of $\F_q^\times$, when needed, can be found by exhaustive search in time polynomial in $q$. We do not attempt to optimize the dependence on $q$, as our main regime of interest is that of small $q$ and much larger ambient dimension.} lossless rank \emph{extractors} over a finite field $\F_q$, i.e., the case $r=t$ and $\F=\F_q$ in Definition \ref{def:rank-condenser}. This setting is crucial both for some of the applications discussed above and for our construction of strong $s$-blocking sets, discussed below.

From a non-explicit perspective, the probabilistic method shows that $(k,r,L)$ strong lossless rank extractors of small size exist over any finite field $\F_q$, with $L=\poly(k)$ (and in fact $L=O(r(k-r))$; see \cref{thm:existence}). However, even for weak lossless rank extractors, known explicit constructions such as \cite{GR08,FS12} achieve such parameters only when the field size is sufficiently large. 

To avoid trivialities, we assume $L<n$, since otherwise one could allow $\rank(E_i M)<r$ for all $i\in[n]$. Under this assumption, the constructions in \cite{GR08,FS12} require $q \ge n$, and hence $q > L$. More concretely, this gives $q > kr^2$ in \cite{GR08}, and $q>kr-\binom{r+1}{2}$ for the original analysis of \cite{FS12}, and $q>r(k-r)$ after the refinement of \cite{FSS14}. In addition, the construction in \cite{FS12} requires $q>k$, as it relies on the existence of an element $\gamma\in\F_q^\times$ whose multiplicative order is at least $k$.

A natural approach to constructing lossless rank extractors over small fields is via concatenation or field-reduction techniques. Indeed, \cite[Proposition~8.5]{FG15} shows that a $(k,r,t,L)$ strong (resp. weak) lossless rank condenser over $\F_{q^d}$ can be converted into a $(k,r,dt,L)$ strong (resp. weak) lossless rank condenser over $\F_q$. However, unless $d=1$, this transformation does not preserve the lossless rank \emph{extractor} property, as it necessarily increases the output dimension to $dt > r$.

This motivates the problem of constructing explicit lossless rank extractors over small fields. The regime of interest to us is one in which the rank parameter $r$ is much smaller than the ambient dimension $k$, while the field size is small and, ideally, bounded independently of $k$. This is also the regime relevant to our applications to strong blocking sets.

In particular, we ask:

\begin{question}\label{question:RC-small-field}
Suppose $r \ll k$. Do there exist explicit $(k,r,L)$ strong or weak lossless rank extractors of size $n > L$ over $\F_q$ such that $L=\poly(k,r)$ and $q$ depends only on $r$?
\end{question}

As we will see, we answer \cref{question:RC-small-field} in the affirmative even for strong lossless rank extractors.

\paragraph{Subspace designs.}
Informally, a \emph{subspace design} is a collection of linear subspaces arranged so that every subspace of a given dimension intersects the family only in a limited way --- either by intersecting only few members (weak designs) or by having small total intersection dimension (strong designs).

\begin{definition}[Subspace designs]
    A collection of subspaces $\{V_i\}_{i\in[n]}\subseteq\mathbb{F}_q^k$ is an $(r,A)$ subspace design if for any subspace $W\subseteq\mathbb{F}_q^k$ of dimension $r$, we have
    \[\sum_{\substack{i\in [n]}}{\bf 1}[V_i\cap W\neq \{0\}]\leq A\quad \text{(weak design)}\qquad \text{or}\qquad\sum_{\substack{i\in [n]}}\dim(V_i\cap W)\leq A \quad \text{(strong design)},\]
where ${\bf 1}[V_i\cap W\neq \{0\}]$ is the indicator function of the event $V_i\cap W\neq \{0\}$.
\end{definition}

Introduced by Guruswami and Xing~\cite{guruswami2013list}, subspace designs have since become a central combinatorial primitive in coding theory and beyond. While the probabilistic method yields subspace designs with near-optimal size and dimension~\cite{guruswami2013list,guruswami2016explicit}, obtaining explicit constructions is substantially more challenging. Guruswami and Kopparty~\cite{guruswami2016explicit} gave the first explicit construction of even strong subspace designs with near-optimal parameters over large fields, based on folded Reed–Solomon and univariate multiplicity codes, thereby fully derandomizing the code constructions of~\cite{guruswami2013list}. Subsequently, Guruswami, Xing, and Yuan~\cite{guruswami2018subspace} constructed strong subspace designs over small finite fields using algebraic function
fields. In the notation where $r$ is the dimension of the test subspace and $t$ is the codimension of each member of the family, their constructions achieve
$A =
O\left(
\frac{rk\log_q k}{t-r+1}
\right)$
for infinitely many ambient dimensions $k$. In particular, their results
include the case $t=r$, giving $A=O(rk\log_q k)$.
They also obtain the stronger bound
$A =
O_c\left(
\frac{rk}{t-r+1}
\right)$
in a more restricted parameter regime in which $k=O(q^c)$, where $c>0$ is an arbitrary constant. Thus, over
fields whose size is independent of $k$, their construction incurs an
additional factor of $\log_q k$. Our results remove this factor in the
case $t=r$ in the parameter regimes covered by our constructions.

Subspace designs play a fundamental role in the construction of list-decodable codes. Many of the best-known constructions---ranging from Reed–Solomon and folded Reed–Solomon codes to algebraic geometry, multiplicity, and rank-metric codes---arise by combining classical algebraic codes with subspace design-based pruning~\cite{guruswami2013list,guruswami2016explicit,GX22}. More recently, this paradigm has been substantially strengthened: the work of Chen and Zhang~\cite{chen2025explicit} shows that explicit folded Reed–Solomon and multiplicity codes themselves achieve optimal list-decodability, and in fact identifies subspace designs as the underlying mechanism governing this phenomenon. In particular, any code admitting sufficiently strong subspace designs---termed \emph{subspace designable codes} \cite{chen2025explicit} or \emph{subspace design codes} \cite{Goyal2025OptimalPG}---automatically enjoys near-optimal list-decodability.

Very recently, Goyal, Guruswami, and Hsieh~\cite{GGH26} constructed constant-alphabet subspace-design codes using the Alon--Edmonds--Luby framework. The kernels of the coordinate maps of their additive codes form strong subspace designs with near-optimal intersection bounds. Their results are most relevant when the codimension of these kernels is substantially larger than the dimension of the test subspace. Exploiting this additional codimension can be important in applications \cite{chen2025explicit,Goyal2025OptimalPG,GGH26,CGGRT26}.
By contrast, we focus on the case where these two parameters are equal, which is the case needed for our construction of strong blocking sets.

Beyond coding theory, strong subspace designs have emerged as a central tool in pseudorandomness and linear-algebraic expansion. In particular, they enable explicit constructions of constant-degree lossless dimension expanders~\cite{guruswami2021lossless}, which can be viewed as a linear-algebraic analogue of expander graphs.
More recently, influenced by the work of Chen and Zhang \cite{chen2025explicit}, a sequence of works~\cite{brakensiek2025random,brakensiek2025combinatorial,Goyal2025OptimalPG} has uncovered a significantly broader and more structural role for subspace designs. These works position strong subspace designs as a unifying bridge between random and explicit constructions, leading to systematic derandomization frameworks and new connections to local properties, matroid theory, and proximity gaps---a central notion in modern proof systems. 

While most of the aforementioned applications primarily rely on strong subspace designs, explicit constructions of weak subspace designs are equally important. In particular, as shown in \cite{FZ14,FZ16}, weak subspace designs are sufficient to construct strong blocking sets, a prominent object in finite geometry that we discuss in the next paragraph. 
More specifically, recent work~\cite{brakensiek2025random} (see also \cite{FZ16,Guo24}) establishes lower bounds on the parameters of weak subspace designs. In particular, it shows that for any collection of subspaces $\{V_i\}_{i\in[n]} \subseteq \overline{\mathbb{F}}^k$ of dimension $r$ (where $\overline{\mathbb{F}}$ is an algebraically closed field), if it forms a $(k-r, A)$ subspace design, then necessarily $A \ge r(k-r)$.
For sufficiently large fields, explicit constructions due to Guruswami and Kopparty~\cite{guruswami2016explicit} match this lower bound even for strong subspace designs. In contrast, obtaining comparable explicit constructions even for weak subspace designs over small fields remained open, as discussed below, and is of particular importance for applications such as constructing strong blocking sets.

In the following question, we focus on the case in which the codimension of
each member of the family equals the dimension of the test subspace.

\begin{question}\label{q2}
    Can one explicitly construct a collection of subspaces $\{V_i\}_{i\in[n]} \subseteq \mathbb{F}^k$ of dimension $r$ (resp. dimension $k-r$) over fields of size $|\mathbb{F}| = O_r(1)$ that forms a $(k-r,A)$ (resp. $(r, A)$) strong or weak subspace design with $A = O(r(k-r))$, or with $A$ close to this lower bound?
\end{question}

As we will see, our results give a positive answer to Question~\ref{q2} even for strong subspace designs, achieving $q=\poly(r)$ and $A=O(r(k-r))$ in the regimes covered by our constructions.

Taken together, all these developments demonstrate that both strong and weak subspace designs are not merely a technical ingredient in code constructions, but rather a fundamental combinatorial primitive with far-reaching applications across coding theory, pseudorandomness, and algebraic combinatorics.
For a more fine-grained historical overview, we refer the reader to the recent survey by Santonastaso and Zullo~\cite{santonastaso2023subspace}.

\paragraph{Blocking sets.} Blocking sets, introduced by Richardson~\cite{Ric56}, are classical objects in finite geometry that capture the minimal structure required to intersect all subspaces of a given codimension. For $1 \le s \le k$, an \emph{affine $s$-blocking set} is a subset $B \subseteq \F_q^k$ that intersects every affine subspace of codimension $s$. In the projective setting, for $1 \le s \le k-1$, a set $B \subseteq \PG(k-1,q)$ is a \emph{strong $s$-blocking set} if for every codimension-$s$ projective subspace $\Sigma$ of $\PG(k-1,q)$, the intersection $B \cap \Sigma$ spans $\Sigma$, i.e., the smallest projective subspace containing $B\cap\Sigma$ equals $\Sigma$. Let $b_q(k,s)$ and $b_q^*(k,s)$ denote the minimum sizes of affine and strong blocking sets, respectively.

Blocking sets are related to several well-studied notions in discrete mathematics, theoretical computer science, and coding theory, including vertex covers~\cite{Fur88}, subspace designs~\cite{FZ16,guruswami2016explicit}, trifferent codes~\cite{BDGP24}, and intersecting and minimal codes~\cite{CL85,ABN22,TQLZ21,XKH25}.

The affine and projective variants are tightly related: it was shown in~\cite{BDGP24} that a set $B \subseteq \PG(k-1,q)$ is a strong $s$-blocking set if and only if the union of the corresponding lines through the origin forms an affine $(s+1)$-blocking set in $\F_q^k$. In particular,
\begin{equation}\label{eq:strong-affine-connection}
    (q-1) b_q^*(k,s) +1 \ge b_q(k,s+1).
\end{equation}

The main problem is to determine the values of $b_q(k,s)$ and $b_q^*(k,s)$. For $s=1$, a classical result of Jamison~\cite{Jam77} and Brouwer and Schrijver~\cite{BS78} shows that
\[
b_q(k,1) = (q-1)k+1.
\]
For $s \ge 2$, however, the problem becomes significantly more difficult and is related to hard problems in additive combinatorics such as the density Hales--Jewett theorem and the cap set problem (see~\cite{BDGP24}).

The best known general lower bound is
\begin{equation}\label{eq:b_q(k,s)-lb}
    b_q(k,s) \ge (q^s-1)(k-s+1)+1,
\end{equation}
which follows from a geometric argument~\cite{Bal11}. 

On the other hand, viewing blocking sets as vertex covers in a natural hypergraph, the Lov\'asz--Stein theorem~\cite[Theorem~2.16]{Juk11} yields the upper bound
\[
b_q(k,s) \le q^s \left(1 + \ln\qb{k}{s}\right),
\]
where $\qb{k}{s}$ denotes the Gaussian binomial coefficient, i.e., the number of $s$-dimensional subspaces of $\F_q^k$. A recent result of Bishnoi et al.~\cite{BDGP24} improves this to an upper bound that is roughly within a factor $O(s)$ of~\eqref{eq:b_q(k,s)-lb} via a probabilistic argument.

For strong $s$-blocking sets, combining~\eqref{eq:strong-affine-connection} and~\eqref{eq:b_q(k,s)-lb} yields
\begin{equation}\label{eq:b_q^*(k,s)-lb}
    b_q^*(k,s) \ge \frac{(q^{s+1}-1)(k-s)}{q-1}.
\end{equation}
For $s=1$, Bishnoi et al.~\cite{BDGP24} obtained a stronger asymptotic lower bound
\[
b_q^*(k,1) \ge (c_q - o(1))(q+1)(k-1)
\]
for some constant $c_q > 1$, using linear programming bounds from coding theory. The same work also gives a probabilistic upper bound for $b_q^*(k,s)$ which is roughly within a factor $O(s)$ of~\eqref{eq:b_q^*(k,s)-lb} (see~\cite[Remark~3.4]{BDGP24}).

The main challenge is to obtain explicit constructions of strong\footnote{This automatically yields explicit constructions for affine $(s+1)$-blocking sets.} $s$-blocking sets with comparable parameters. When the field size $q$ is sufficiently large as a function of $k$, constructions based on higgledy-piggledy subspaces~\cite{FZ14,FZ16} yield strong $s$-blocking sets of size $O_{k,s}(q^s)$. In the regime where $q,s$ are fixed and $k$ grows, the best known explicit constructions of strong $s$-blocking sets have size $O_s(q^s k)$~\cite{ABDN24,BT24}.

For $s=1$, a breakthrough result of Alon, Bishnoi, Das, and Neri~\cite{ABDN24} gave the first explicit construction of size $O(qk)$ using expander graphs and coding-theoretic tools. This approach was later extended by Bishnoi and Tomon~\cite{BT24} to all $s \ge 2$, yielding constructions of size at most $C_s q^s k$, where $C_s = 2^{O(s^2 \log s)}$. However, the dependence on $s$ in $C_s$ is large, and their construction requires the field size $q$ to be at least exponential in $s$. For smaller fields, an explicit construction of size $q^{O(s^2)}k$ was given in \cite{BT24}. 



Our main contribution is explicit constructions of strong $s$-blocking sets of size $O(s(k-s) q^s)$ over fields $\F_q$ with $q\geq \poly(s)$, together with extensions to smaller fields. This improves the $s$-dependence in the leading constant of the construction in~\cite{BT24} from exponential to polynomial. By \eqref{eq:strong-affine-connection}, these constructions also yield explicit affine $(s+1)$-blocking sets with comparable parameters.

\subsection{Main Results}

In this paper, we obtain new explicit constructions of lossless rank extractors, subspace designs, and strong $s$-blocking sets over small fields, achieving parameters that were previously unknown.

\paragraph{Lossless rank extractors.} Over non-prime finite fields, we obtain the following result on explicit strong lossless rank extractors.

\begin{theorem}[Informal version of \cref{thm:fullrankbound}]\label{thm:rankboundinformal}
For every $r\geq 1$ and every sufficiently large non-prime prime power $q\geq \poly(r)$, there exist infinitely many $k\geq r$ such that one can explicitly construct a $(k,r,L)$ strong lossless rank extractor of size $n$ over $\F_q$ with $L=O(r(k-r))$ and $L/n\leq q^{-1/4}$.
\end{theorem}

Note that the field size $q$ is only required to be polynomial in $r$, and may be independent of $k$. The upper bound on the number of bad matrices is $L=O(r(k-r))$, which matches the upper bound $r(k-r)$ in \cite{FSS14,For14} up to a constant factor. 

\begin{remark}
\cref{thm:rankboundinformal} and the subsequent theorems hold only for infinitely many values of $k$, for the same reason that families of algebraic geometry (AG) codes constructed from function field towers are known to be asymptotically good only for infinitely many lengths, rather than for all lengths. This “infinitely many $k$” phenomenon also appears in the strong $1$-blocking set construction of Alon, Bishnoi, Das, and Neri~\cite{ABDN24}, which uses either Justesen codes or AG codes, and hence yields asymptotically good codes only for infinitely many lengths. Similarly, it also appears in the explicit constructions of Bishnoi and Tomon~\cite{BT24} over small fields, which are also based on AG codes.

One can extend \cref{thm:rankboundinformal} and the subsequent theorems to all $k$, and indeed we prove such versions for all $k$, but then the bound $L=O(r(k-r))$ must increase by an additional factor of at most $O(q)$. See \cref{thm:fullrankbound}, for example.
\end{remark}

Over prime finite fields, we obtain the following weaker result.

\begin{theorem}[Informal version of \cref{thm:rankprime}]\label{thm:rankprime-informal}
For every $r\ge 1$ and $\delta\in (0,1)$, and every sufficiently large prime $q\ge \poly(r/\delta)$, there exist infinitely many $k\ge r$ such that an explicit $(k,r,L)$ strong lossless rank extractor of size $n$ over $\F_q$ can be constructed with
$L\le (2r/\delta)^{O(\log r)} r(k-r)$ and
$L/n\le \delta$.
\end{theorem}

Over fields of absolute constant size, such as $\F_2$, we do not know how to explicitly construct lossless rank extractors for which, for every full-rank matrix $M$, most $E_i\in\mathcal{E}$ satisfy $\rank(E_i M)=r$. However, we can explicitly construct a weaker object in which \emph{some} $E_i$ satisfies $\rank(E_i M)=r$. Indeed, Forbes \cite{For14} defined rank condensers in this weaker sense.

We restrict to the case where the rank equals the output dimension and call these objects \emph{rank dispersers}, drawing on the intuition that dispersers are one-sided weakenings of extractors.

\begin{definition}[Lossless rank dispersers]\label{defi:disperser}
Let $1\leq r\le k$ be integers. A finite collection $\mathcal{E}=\{E_i\}_{i=1}^n \subseteq \F^{r\times k}$ of matrices over a field $\F$ is called a $(k,r)$ lossless rank disperser over $\F$ if for every matrix $M \in \F^{k \times r}$ of rank $r$, there exists $i\in [n]$ such that $\rank(E_i M) = r$.
We call $n$ the size of $\mathcal{E}$.
\end{definition}

In other words, $\mathcal{E}$ is a $(k,r)$ lossless rank disperser of size $n$ if and only if it is a $(k,r,n-1)$ weak lossless rank extractor of size $n$.

The following theorem gives explicit lossless rank dispersers over \emph{any} finite field, including $\F_2$.

\begin{theorem}[Informal version of \cref{thm:rcsmall}]\label{thm:rcsmall-informal}
For every $r\geq 1$ and every prime power $q>1$, there exist infinitely many $k\geq r$ such that one can explicitly construct a $(k,r)$ lossless rank disperser of size $n$ over $\F_q$, where $n=O\left(\max\left\{\frac{c\log r}{\log q},2\right\}^r r(k-r)\right)$ and $c>0$ is some absolute constant.
\end{theorem}

Compared with \cref{thm:rankboundinformal} and \cref{thm:rankprime-informal}, \cref{thm:rcsmall-informal} only yields lossless rank dispersers. However, as we will see, this is still sufficient for constructing strong $s$-blocking sets.

We also remark that the constant $2$ in $\max\left\{\frac{c\log r}{\log q},2\right\}$ in \cref{thm:rcsmall-informal} (and similarly in \cref{thm:design-constant-informal} and \cref{thm:blocking-constant-informal}) can be improved to $1$ when $q$ is non-prime; see the proof of \cref{thm:rcsmall} and \cref{remark:two}. However, this improvement is relevant only for sufficiently large non-prime fields $q \ge \poly(r)$, in which case one can instead apply \cref{thm:rankprime-informal}, which yields stronger bounds.

\paragraph{Subspace designs.}

On non-prime, prime, and all finite fields, we obtain explicit strong and weak subspace designs with various parameters.

\begin{theorem}[Informal version of \cref{thm:designbound}\,\eqref{item:nonprime}]\label{thm:design-nonprime-informal}
For every $r\ge 1$ and every sufficiently large non-prime prime power $q\ge \poly(r)$, there exist infinitely many $k\ge r$ such that one can explicitly construct an $(r,A)$ (resp. $(k-r, A)$) strong subspace design of size $n$ over $\F_q$, consisting of $(k-r)$-dimensional (resp. $r$-dimensional) subspaces of $\F_q^k$, with $A=O(r(k-r))$ and $n\ge  q^{1/4} A$.
\end{theorem}

\begin{theorem}[Informal version of \cref{thm:designbound}\,\eqref{item:prime}]\label{thm:design-prime-informal}
For every $r\ge 1$, $\delta\in (0,1)$, and every sufficiently large prime $q\ge \poly(r/\delta)$, there exist infinitely many $k\ge r$ such that one can explicitly construct an $(r,A)$ (resp. $(k-r, A)$) strong subspace design of size $n$ over $\F_q$, consisting of $(k-r)$-dimensional (resp. $r$-dimensional) subspaces of $\F_q^k$, with $A\leq (2r/\delta)^{O(\log r)} r(k-r)$ and $n\ge A/\delta$.
\end{theorem}

\begin{theorem}[Informal version of \cref{thm:designbound}\,\eqref{item:small}]\label{thm:design-constant-informal}
For every $r\ge 1$ and every prime power $q>1$, there exist infinitely many $k\ge r$ such that one can explicitly construct an $(r,A)$ (resp.\ $(k-r,A)$) weak subspace design of size $n$ over $\F_q$, consisting of $(k-r)$-dimensional (resp.\ $r$-dimensional) subspaces of $\F_q^k$, with $A=n-1$ and
$n=O\!\left(\max\left\{\frac{c\log r}{\log q},2\right\}^r r(k-r)\right)$
for some absolute constant $c>0$.
\end{theorem}

\begin{remark}
While explicit subspace designs over small finite fields are known in the literature, no explicit constructions achieving \cref{thm:design-nonprime-informal,thm:design-prime-informal,thm:design-constant-informal} were known prior to this work, even for weak subspace designs. 

One possible approach is to first construct an explicit $(r,A)$ weak subspace design in $\F_{q^d}^k$ over a large finite field $\F_{q^d}$ with $q^d>k$, consisting of codimension-$r$ subspaces, and then use \cite[Lemma~8]{guruswami2016explicit} to obtain an $(r,A)$ weak subspace design in $\F_q^{dk}$. 

However, this transformation increases the codimension of the subspaces to $dr>r$. Moreover, when $r$ is constant and $q=O_r(1)=O(1)$, we have $d=O(\log_q k)=O(\log k)$, implying that the codimension of the subspaces becomes logarithmic in $k$ rather than constant. 
\end{remark}


\paragraph{Strong $s$-blocking sets.}

Finally, we obtain explicit strong $s$-blocking sets of improved size over small fields. By \eqref{eq:strong-affine-connection}, these also yield explicit affine $(s+1)$-blocking sets over small fields. See \cref{tab:sbs-compare} for a summary of our results and prior work.

\begin{table}[htp]
    \centering
    \begin{tabular}{Sc|Sc|Sc}
    \hline
                           Reference   & Condition on $q$   & Upper bound on the size
                              \\ \hline\hline
     \cite[Theorem~16]{BT24} &  $q$ is square and $\geq s^{\Omega(s)}$   &  $2^{O(s^2\log s)}kq^s$    \\\hline
     \cite[Theorem~17]{BT24}  & $-$                  &   $kq^{O(s^2)}$ \\\hline
     \cref{thm:blocking-nonprime-informal} &  $q$ is non-prime and $\geq \poly(s)$  & $O(s(k-s)q^s)$ \\\hline
     \cref{thm:blocking-prime-informal}  & $q$ is prime and $\geq \poly(s)$  & $(2s)^{O(\log s)} (k-s) q^s$ \\\hline
     \cref{thm:blocking-constant-informal} & $-$ & $O\!\left(\max\left\{\frac{c\log (2s)}{\log q},2\right\}^{s+1} s (k-s) q^s \right)$ \\\hline
       \cref{thm:blocking-eps-informal} & $-$ & $k q^{(2+o_{qs}(1))s}$ \\ \hline
    \end{tabular}
    \caption{A comparison of explicit constructions of strong $s$-blocking sets}
    \label{tab:sbs-compare}
\end{table}

The following theorem shows that for every sufficiently large non-prime $q\ge \poly(s)$, there exist explicit strong $s$-blocking sets of size $O(s(k-s)q^s)$, matching the best-known non-explicit upper bound (see~\cite[Remark~3.4]{BDGP24}) up to an absolute constant factor.

\begin{theorem}[Informal version of \cref{cor:blocking}\,\eqref{item:blocking1}]\label{thm:blocking-nonprime-informal}
For every $s\ge 1$ and every sufficiently large non-prime prime power $q\ge \poly(s)$, there exist infinitely many $k>s+1$ such that one can explicitly construct a strong $s$-blocking set $B\subseteq \PG(k-1,q)$ of size $O(s(k-s) q^s)$.
\end{theorem}

When $q\ge \poly(s)$ is prime, we obtain a weaker size bound with a quasi-polynomial factor in $s$.

\begin{theorem}[Informal version of \cref{cor:blocking}\,\eqref{item:blocking2}]\label{thm:blocking-prime-informal}
For every $s\ge 1$ and every sufficiently large prime $q\ge \poly(s)$, there exist infinitely many $k>s$ such that one can explicitly construct a strong $s$-blocking set $B\subseteq \PG(k-1,q)$ of size at most $(2s)^{O(\log s)} (k-s) q^s$.
\end{theorem}

For arbitrary $q$, we obtain the following result.

\begin{theorem}[Informal version of \cref{cor:blocking}\,\eqref{item:blocking3}]\label{thm:blocking-constant-informal}
For every $s\ge 1$ and every prime power $q>1$, there exist infinitely many $k>s$ such that one can explicitly construct a strong $s$-blocking set $B\subseteq \PG(k-1,q)$ of size
$O\!\left(\max\left\{\frac{c\log (2s)}{\log q},2\right\}^{s+1} s (k-s) q^s \right)$,
where $c>0$ is an absolute constant. In particular, when $q=2$, the size of $B$ is at most $2^{O(s\log\log s)}(k-s)q^s$.
\end{theorem} 

Setting $s=1$, \cref{thm:blocking-constant-informal} recovers the main result of Alon, Bishnoi, Das, and Neri that there exist explicit strong $1$-blocking sets of size $O(k q)$. Notably, our proof (and construction) does not use expander graphs as in \cite{ABDN24}, and can be viewed as a purely algebraic approach to the same result.

\cref{thm:blocking-nonprime-informal,thm:blocking-prime-informal,thm:blocking-constant-informal} are proved using algebraic techniques, including tools from function fields and polynomial identity testing. 

We also develop a Fourier-analytic approach, reducing the construction of explicit strong $s$-blocking sets in $\PG(k-1,q)$ to that of explicit $\eps$-biased sets in $\F_q^k$. By instantiating this reduction with Ta-Shma's breakthrough construction~\cite{Ta17} (and its extension to general finite fields $\F_q$ by Jalan and Moshkovitz~\cite{JM21}), we obtain the following theorem.

\begin{theorem}[Informal version of \cref{prop:explicit-biased-to-explicit-blocking}]\label{thm:blocking-eps-informal}
For every $1\le s<k$ and every prime power $q>1$, there exists an explicit strong $s$-blocking set $B\subseteq \PG(k-1,q)$ of size at most $k q^{(2+o_{qs}(1))s}$.
\end{theorem}

Previously, \cite[Theorem~17]{BT24} constructed explicit strong $s$-blocking sets of size $k q^{O(s^2)}$ for all $q$. 
\cref{thm:blocking-eps-informal} improves this to a size bound of $k q^{(2+o_{qs}(1))s}$.

While this bound is weaker than that of \cref{thm:blocking-constant-informal} for large enough $q$, e.g., for $q\ge \poly(s)$, it is the best currently known bound when $q=O(1)$, namely $2^{O(s)} k q^s$. In contrast, \cref{thm:blocking-constant-informal} yields $2^{O(s\log\log s)} k q^s$ in this regime. Moreover, \cref{thm:blocking-eps-informal} applies to all $k$, rather than only infinitely many $k$.




\subsection{Technical Overview} 

We primarily focus on \cref{question:RC-small-field}, namely the problem of explicitly constructing $(k,r,L)$ lossless rank extractors over finite fields whose size depends only on $r$ (and in fact $\poly(r)$), in the regime where $r\ll k$. We also explain why some existing constructions do not directly address this question. 

For simplicity, we focus on weak lossless rank extractors. The guarantee for strong lossless rank extractors follows from the same argument by counting zeros with multiplicity.

For a field $\F$ with $|\F|>k$, Forbes and Shpilka~\cite{FS12} constructed an explicit $(k,r,L)$ weak lossless rank extractor $\mathcal{E}$ over $\F$ with $L=rk-\binom{r+1}{2}$ (later improved to $r(k-r)$ in~\cite{FSS14,For14}). Their proof uses the polynomial method and proceeds as follows. First, they construct the symbolic matrix $E(x)=(E_{ij}(x))_{i\in [r], j\in [k]}$ with entries
\[
E_{ij}(x)=(\gamma^{i-1}x)^{j-1},
\]
where $\gamma\in\F^\times$ is chosen such that $1,\gamma,\gamma^2,\dots,\gamma^{k-1}$ are distinct. They then show that for any full-rank matrix $M\in \F^{k\times r}$, we have $\det(E(x)M)\neq 0$ symbolically. It is also straightforward to bound the degree of the polynomial $\det(E(x)M)\in\F[x]$ by $rk-\binom{r+1}{2}$. Thus, by the polynomial method, we have $\det(E(a)M)\neq 0$, or equivalently $\rank(E(a)M)=r$, for all but at most $rk-\binom{r+1}{2}$ elements $a\in\F$. It therefore suffices to define $\mathcal{E}=\{E(a): a\in S\}$ for a finite set $S\subseteq \F$ of size greater than $L$.

One reason this construction requires a large field size is that the analysis uses Vandermonde determinants associated with elements in the set $\{1,\gamma,\gamma^2,\dots,\gamma^{k-1}\}\subseteq\F^\times$, and these elements must be distinct for the Vandermonde determinants to be nonzero. In particular, this forces $|\F|>k$.

Our first observation is that, since we are already working with the polynomial ring $\F[x]$, it is not necessary to use distinct values from the field $\F$; instead, we can use elements of $\F[x]$. Indeed, an earlier construction of Gabizon and Raz~\cite{GR08} also uses the polynomial method, but defines the symbolic matrix $E=(E_{ij})$ by $E_{ij}=x^{ij}$. This construction uses no constants from $\F$ other than $1$, and instead relies on monomials of different degrees. Roughly speaking, this suggests that we can trade distinct constants for elements of different degrees.

However, another issue remains. If we directly adopt the construction of Gabizon and Raz over a small field $\F_q$, then the resulting polynomial $\det(E(x)M)$ has degree much larger than $q$, and may vanish at every evaluation point in $\F_q$.

\paragraph{Function fields.} 
To address this issue, we draw inspiration from algebraic geometry (AG) codes, whose analysis can be viewed as a function-field generalization of the polynomial method. Geometrically, the idea is to replace a line by an \emph{algebraic curve}, which, even over a small field $\F_q$, can have many evaluation points $N\gg q$. Algebraically, this corresponds to working in a function field. We adopt this approach in this paper.

Specifically, we show that a function-field analog of the Gabizon--Raz construction~\cite{GR08} can be carried out over finite fields of size (at least) $\poly(r)$. This already implies a weaker version of \cref{thm:rankboundinformal}, although the upper bound $L$ on the number of bad matrices is $O(r^2 k)$, as in~\cite{GR08}, rather than $O(rk)$. This construction and its analysis are presented in \cref{sec:basic}. 

We then revisit the Forbes--Shpilka construction~\cite{FS12} and show that a function-field analog can also be carried out over finite fields of size (at least) $\poly(r)$, leading to a full proof of \cref{thm:rankboundinformal}. The idea is that when $q\ll k$, although we cannot find $\gamma\in\F_q^\times$ such that $1,\gamma,\dots,\gamma^{k-1}$ are all distinct, we can instead use elements of the form
\[
g_1,\gamma g_1,\dots,\gamma^{q-2} g_1,\;
g_2,\gamma g_2,\dots,\gamma^{q-2} g_2,\;
\cdots,\;
g_h,\gamma g_h,\dots,\gamma^{q-2} g_h,
\]
for suitable transcendental elements $g_1,\dots,g_h$ in the function field, where $h\approx \frac{k}{q-1}$ and $\gamma$ is a generator of $\F_q^\times$. The use of these transcendental elements increases the degree of $\det(E(x)M)$. However, for sufficiently large $q\ge \poly(r)$ and carefully chosen $g_1,\dots,g_h$, this increase is relatively small and dominated by other terms. This construction and its analysis are presented in \cref{sec:FS}.

As a side remark, while it is well known that algebraic curves and function fields underlie AG codes, and prior works such as \cite{ABDN24,BT24} also use AG codes, our proof does not use AG codes as a black box. Instead, we directly carry out analogs of the analyses in~\cite{GR08,FS12} using the function-field generalization of the polynomial method.

\paragraph{Field reduction.}

Next, we use a field-reduction technique to obtain explicit lossless rank dispersers over very small fields such as $\F_2$, thereby proving \cref{thm:rcsmall-informal}. To do so, we first choose an extension field $\F_Q$ of $\F_q$ of size $\poly(r)$, and apply \cref{thm:rankboundinformal} to obtain an explicit lossless rank extractor (and hence disperser) $\mathcal{E}$ over $\F_Q$. 

We then transform each matrix $E_i\in\F_Q^{r\times k}$ into a collection $S(E_i)$ of matrices over $\F_q$. Fix a basis $e_1,\dots,e_d$ of $\F_Q$ over $\F_q$. Each row $u$ of $E_i$ can be written as a linear combination of $d$ row vectors $u_1,\dots,u_d$ over $\F_q$, and we replace $u$ by one of the $u_j$ for $j\in[d]$. Since there are $d$ choices for each row, this yields $d^r=2^{\widetilde{O}(r)}$ matrices over $\F_q$, which form the collection $S(E_i)$. We then show that $\bigcup_{E_i\in\mathcal{E}} S(E_i)$ is a lossless rank disperser over $\F_q$, yielding \cref{thm:rcsmall-informal}.

\paragraph{Addressing prime fields via PIT.} 

Note that \cref{thm:rankboundinformal} requires $q$ to be non-prime. This is due to the need to choose a suitable function field. Specifically, while a function field $F$ provides many evaluation points $N(F)\gg q$, it also introduces an error term depending on its genus $g(F)$. One therefore wants $N(F)/g(F)$ to be large, so that the benefit outweighs the cost.

When $q$ is non-prime, explicit function fields with $N(F)/g(F)\geq q^c$ for some constant $c>0$ are known. In particular, when $q$ is square, we use the Garcia--Stichtenoth tower~\cite{GS96}, which attains the optimal Drinfeld--Vl\u{a}du\c{t} bound $\sqrt{q}-1$~\cite{VS83}. When $q=p^d$ with $d>1$ odd, we instead use the Bassa--Beelen--Garcia--Stichtenoth tower~\cite{BBGS15}, which similarly guarantees $N(F)/g(F)\geq q^c$.

This leaves the case where $q\geq \poly(r)$ is prime. In this case, Serre~\cite{Serre} used class field towers to show the existence of infinite families of function fields $F$ over $\F_q$ with $N(F)/g(F)\geq c\log q$ for some constant $c>0$. This result was suggested in~\cite{BT24} as a basis for constructing strong $s$-blocking sets. However, class field towers are generally not regarded as explicit objects; see, e.g.,~\cite{S01}.

We develop a new approach, based on polynomial identity testing (PIT), that reduces the prime case to the non-prime prime power case, and thereby yields explicit lossless rank extractors over $\F_q$ for prime $q\geq \poly(r)$. We begin by taking $Q=q^2$ and constructing an explicit lossless rank extractor $\mathcal{E}$ over $\F_Q$ via \cref{thm:rankboundinformal}. Starting from  $\mathcal{E}$, one can apply the field-reduction procedure described above. In this procedure, each row $u\in\F_Q^k$ is written as $u=u_1 e_1+u_2 e_2$ with respect to a basis $\{e_1,e_2\}$ of $\F_Q/\F_q$, and one replaces $u$ by either $u_1$ or $u_2$. As the choices are made independently for each of the $r$ rows, a direct application of field reduction increases the size by a factor $[\F_Q:\F_q]^r=2^r$.

To avoid this exponential blow-up, let $u^{(i)}\in \F_Q^k$ denote the $i$-th row. 
Write
\[
u^{(i)} = u^{(i)}_1 e_1 + u^{(i)}_2 e_2,
\]
where $u^{(i)}_1,u^{(i)}_2 \in \F_q^k$. 
Instead of choosing between $u^{(i)}_1$ and $u^{(i)}_2$, we replace $u^{(i)}$ by a linear combination 
$a^{(i)} u^{(i)}_1 + b^{(i)} u^{(i)}_2$, where $a^{(i)}, b^{(i)} \in \F_q$. For $q\geq \poly(r)$, a random choice of these coefficients is likely to preserve the nonzeroness of the determinant. The problem is therefore to construct a small explicit set of such choices that contains many good ones.

We show that this problem reduces to deterministic black-box PIT for a class of polynomials of the form $\det(A(x))$, where $A(x)=\sum_i A_i x_i$ and each $A_i$ has rank at most one. Gurjar and Thierauf~\cite{DBLP:journals/cc/GurjarT20} showed that this class admits quasi-polynomial-size hitting sets. We adapt their construction to ensure that evaluation points lie in $\F_q^k$ for sufficiently large $q\geq \poly(r)$, yielding \cref{thm:rankprime-informal}. Our method can be viewed as a PIT-based field-reduction technique that reduces the exponential blow-up to quasi-polynomial.

\paragraph{Subspace designs and strong $s$-blocking sets.}

Our results on subspace designs (\cref{thm:design-nonprime-informal,thm:design-prime-informal,thm:design-constant-informal}) follow from \cref{thm:rankboundinformal,thm:rankprime-informal,thm:rcsmall-informal} via the equivalence between lossless rank extractors and subspace designs (\cref{lem:rc2design}; see also \cite[Proposition~6.1]{FG15}).


Our results on strong $s$-blocking sets (\cref{thm:blocking-nonprime-informal,thm:blocking-prime-informal,thm:blocking-constant-informal}) use a similar connection, via a reduction to the weak subspace design/lossless rank disperser property (\cite[Proposition~10]{FZ16}; see also \cref{thm:lre-to-sbs}).

Finally, \cref{thm:blocking-eps-informal} provides a different construction based on $\eps$-biased sets and Fourier analysis. The key point is that $\eps$-biased sets are pseudorandom against functions whose Fourier transform has small $L_1$ norm, while the strong $s$-blocking set property can be expressed using indicator functions of subspaces and their differences, which have small $L_1$ norms. This gives a reduction from constructing explicit strong $s$-blocking sets to constructing explicit $\eps$-biased sets. Plugging in the explicit constructions of $\eps$-biased sets from \cite{Ta17,JM21} then yields explicit strong $s$-blocking sets.

\section{Preliminaries and Notation}

Define $[n] := \{1,\dots,n\}$. 
Denote by $\Sym(n)$ the symmetric group on $[n]$, i.e., the group of all permutations of $[n]$. Denote by $R^{r\times k}$ the set of $r\times k$ matrices over a commutative ring $R$. For a subspace $V\subseteq \F^k$ over a field $\F$, denote by $V^{\perp}$ its orthogonal complement with respect to the standard inner product. For a matrix $M$, let $\rowspan(M)$ and $\colspan(M)$ denote its row space and column space, respectively. The projective space of dimension $n$ over a finite field $\F_q$, denoted by $\PG(n,q)$, is the set of all one-dimensional subspaces of $\F_q^{n+1}$.

Let $r \le k$. For an $r\times k$ matrix $A$ and a subset $I \subseteq [k]$ of size $r$, let $A_I$ denote the $r\times r$ submatrix of $A$ consisting of the columns indexed by $I$. Similarly, for a $k\times r$ matrix $B$ and a subset $J \subseteq [k]$ of size $r$, let $B^J$ denote the $r\times r$ submatrix of $B$ consisting of the rows indexed by $J$.

We will use the following classical \emph{Cauchy--Binet formula}, which expresses the determinant of a product of rectangular matrices as a sum over products of minors. For a proof, see, e.g., \cite[Appendix~B.3]{For14}.

\begin{lemma}[Cauchy--Binet formula]\label{lem:CB}
Let $A$ be an $r\times k$ matrix and $B$ be a $k\times r$ matrix over a commutative ring, with $r \le k$. Then
\[
\det(AB)
=
\sum_{I \subseteq [k], |I|=r}
\det(A_I)\,\det(B^I).
\]
\end{lemma}

Following \cite{GR08,FS12}, we will also need a statement asserting the unique optimality of a maximal minor of a full-rank matrix.

\begin{lemma}\label{lem:unique}
Let $M \in \F^{r \times k}$ be a full-rank matrix over a field $\F$, with $r \le k$. 
Let $I \subseteq [k]$ be the subset produced by the following procedure.
Initialize $I = \emptyset$. For $i = k, k-1, \dots, 1$, add $i$ to $I$ if the $i$-th column of $M$ does not lie in the $\F$-linear span of the columns indexed by the current set $I$.

Then $|I| = r$ and $M_I$ is nonsingular. Moreover, for every $J \subseteq [k]$ of size $r$ such that $M_J$ is nonsingular, the following hold:
\begin{enumerate}
\item $i_\ell \ge j_\ell$ for all $\ell \in [r]$, where $i_\ell$ and $j_\ell$ denote the $\ell$-th smallest elements of $I$ and $J$, respectively.
\item Let $w_1 > w_2 > \cdots > w_k$ be integers. Then
\[
\sum_{i \in I} w_i \le \sum_{i \in J} w_i,
\]
with equality if and only if $I = J$.
\end{enumerate}
\end{lemma}

\begin{proof}
Let $c_i$ denote the $i$-th column of $M$. By construction, each selected column lies outside the span of the previously selected ones, so the columns indexed by $I$ are linearly independent. If they did not span $\F^r$, then some column of $M$ would lie outside their span, and would have been added when considered, a contradiction. Thus $|I|=r$ and $M_I$ is nonsingular.

Let $I=\{i_1<\cdots<i_r\}$ and $J=\{j_1<\cdots<j_r\}$, where $M_J$ is nonsingular. We prove that $i_\ell \ge j_\ell$ for all $\ell$.

Suppose not, and let $\ell$ be the largest index such that $i_\ell < j_\ell$. Then for all $t > \ell$, we have $i_t \ge j_t > j_\ell$. Hence when the algorithm considers $j_\ell$, the indices already chosen are exactly $i_{\ell+1},\dots,i_r$.

Assume that $j_\ell$ is not selected. Then $c_{j_\ell} \in \Span(c_{i_{\ell+1}},\dots,c_{i_r})$. For each $t > \ell$, the index $j_t$ is processed before $j_\ell$. If $j_t$ is selected, then $j_t \in \{i_{\ell+1},\dots,i_r\}$; otherwise, $c_{j_t}$ lies in the span of $c_{i_{\ell+1}},\dots,c_{i_r}$. In either case,
\[
c_{j_t} \in \Span(c_{i_{\ell+1}},\dots,c_{i_r}).
\]
Thus
\[
\Span(c_{j_{\ell+1}},\dots,c_{j_r})
\subseteq
\Span(c_{i_{\ell+1}},\dots,c_{i_r}).
\]
Both spaces have dimension $r-\ell$, since the corresponding sets of columns are linearly independent. Hence the two spans are equal, and therefore
\[
c_{j_\ell} \in \Span(c_{j_{\ell+1}},\dots,c_{j_r}),
\]
contradicting the linear independence of the columns indexed by $J$. Thus $j_\ell$ must be selected, contradicting $i_\ell < j_\ell$. This proves $i_\ell \ge j_\ell$ for all $\ell$.

Finally, if $w_1 > \cdots > w_k$, then $i_\ell \ge j_\ell$ implies $w_{i_\ell} \le w_{j_\ell}$ for each $\ell\in [r]$, and hence
\[
\sum_{i\in I} w_i \le \sum_{i\in J} w_i.
\]
Equality holds only if $w_{i_\ell} = w_{j_\ell}$ for all $\ell\in [r]$, which (since $w_1 > \cdots > w_k$) implies $i_\ell = j_\ell$ for all $\ell\in [r]$, i.e., $I = J$.
\end{proof}

Recall the definition of lossless rank extractors and that of subspace designs. The following lemmas show a connection between the two notions. 
It is similar to \cite[Proposition~6.1]{FG15} except that we consider both primal and dual forms and require the matrices to have full rank.

\begin{lemma}\label{lem:rc2design}
Let $r\leq k$.
Suppose $\mathcal{E}\subseteq \F^{r\times k}$ is a finite collection of matrices,
each of rank $r$.
Define the collections of subspaces
\[
\mathcal{V}(\mathcal{E})=\{\rowspan(E): E\in\mathcal{E}\}
\quad\text{and}\quad
\mathcal{V}^\perp(\mathcal{E})=\{\rowspan(E)^\perp: E\in\mathcal{E}\}.
\]
Then every subspace in $\mathcal{V}(\mathcal{E})$ has dimension $r$, and every subspace in $\mathcal{V}^\perp(\mathcal{E})$ has dimension $k-r$.

Moreover, $\mathcal{E}$ is a $(k,r,L)$ strong (resp. weak) lossless rank extractor if and only if
$\mathcal{V}(\mathcal{E})$ is a $(k-r,L)$ strong (resp. weak) subspace design, which holds if and only if
$\mathcal{V}^\perp(\mathcal{E})$ is an $(r,L)$ strong (resp. weak) subspace design.
\end{lemma}

\begin{proof}
Define $V_E:=\rowspan(E)$ for $E\in\mathcal E$.
Since each $E\in\mathcal E$ has rank $r$, every $V_E$ has dimension $r$,
and thus every $V_E^\perp$ has dimension $k-r$.

For the main claim, consider a full-rank matrix $M\in\F^{k\times r}$ and let $W=\colspan(M)^\perp$, which has dimension $k-r$.
For $E\in\mathcal E$, consider the linear map $\phi_{E,M}:V_E\to\F^r$ that sends $v\in V_E$ to $vM$. Its image is the row space of $EM$, whose dimension is $\rank(EM)$.
Its kernel is $V_E\cap \colspan(M)^\perp=V_E\cap W$. Therefore, $r-\rank(EM)=\dim(V_E\cap W)$. Summing over $E\in\mathcal{E}$, we obtain
\[
\sum_{E\in\mathcal{E}} (r-\rank(EM))=\sum_{E\in\mathcal{E}} \dim(V_E\cap W).
\]
Similarly, the number of $E\in\mathcal{E}$ satisfying $\rank(EM)<r$ equals the number of $V_E\in\mathcal{V}(\mathcal{E})$ satisfying $V_E\cap W\neq \{0\}$.
As $M$ ranges over all full-rank matrices in $\F^{k\times r}$, $W$ ranges over all $(k-r)$-dimensional subspaces, possibly with repetitions. 
So $\mathcal{E}$ is a $(k,r,L)$ strong (resp. weak) lossless rank extractor if and only if
$\mathcal{V}(\mathcal{E})$ is a $(k-r,L)$ strong (resp. weak) subspace design.

Finally, for any $r$-dimensional subspace $U$ and any $V\in \mathcal V(\mathcal E)$, we have
\begin{align*}
\dim(V^\perp\cap U)&=k-\dim((V^\perp\cap U)^\perp)=k-\dim(V+ U^\perp)\\
&=k-(\dim V+\dim U^\perp-\dim(V\cap U^\perp))
=\dim(V\cap U^\perp).
\end{align*}
As $U$ ranges over all $r$-dimensional subspaces, $U^\perp$ ranges over all $(k-r)$-dimensional subspaces. So $\mathcal{V}(\mathcal{E})$ is a $(k-r,L)$ strong (resp. weak) subspace design iff $\mathcal{V}^\perp(\mathcal{E})$ is an $(r,L)$ strong (resp. weak) subspace design.
\end{proof}

In general, the matrices in a lossless rank extractor need not have full rank.
However, by discarding the rank-deficient matrices, we can ensure full rank at the cost
of losing at most $L$ elements.

\begin{lemma}\label{lem:reduced}
Let $r\le k$. Suppose $\mathcal{E}_0\subseteq \F^{r\times k}$ is a $(k,r,L)$ strong (resp. weak) lossless
rank extractor of size $n\ge L$. Let
\[
\mathcal{E}=\{E\in\mathcal{E}_0 : \rank(E)=r\}.
\]
Then $\mathcal{E}$ is a $(k,r,L)$ strong (resp. weak) lossless rank extractor of size at least $n-L$
consisting only of rank-$r$ matrices. Consequently, $\mathcal{V}(\mathcal{E})$ and $\mathcal{V}^\perp(\mathcal{E})$ as defined in \cref{lem:rc2design} form a $(k-r,L)$ strong (resp. weak) subspace design and an $(r,L)$ strong (resp. weak) subspace design, respectively, both of size at least $n-L$.
\end{lemma}

\begin{proof}
Fix any full-rank $M\in\F^{k\times r}$. If $\rank(E)<r$, then $\rank(EM)\le \rank(E)<r$.
Since $\mathcal{E}_0$ is a $(k,r,L)$ weak lossless rank extractor, there are at most $L$ such matrices $E$ in $\mathcal{E}_0$. Hence $|\mathcal{E}|\ge n-L$.

Removing these matrices does not increase the quantity $\sum_{E}(r-\rank(EM))$ or the number of $E$ satisfying $\rank(EM)<r$ for any $M$,
so $\mathcal{E}$ remains a $(k,r,L)$ strong (resp. weak) lossless rank extractor.
The rest follows from \cref{lem:rc2design}.
\end{proof}

\begin{remark}
Although lossless rank extractors and subspace designs are essentially equivalent, their family sizes are typically optimized in opposite directions. For a rank extractor, one often fixes a target error $\eps$ and seeks to minimize the family size $n$ subject to $L/n \leq \eps$, since $\log n$ is the seed length needed to sample a matrix from the family. Thus, an upper bound on the required family size is desirable. For a subspace design, one instead fixes the intersection bound $L$ and seeks constructions with $n$ as large as possible subject to this bound; thus, a lower bound on the achievable family size is desirable.
\end{remark}






\subsection{Preliminaries on Function Fields}

We recall basic facts about finitely generated function fields of transcendence degree one over a finite field $\F_q$. The theory of such function fields is equivalent to that of nonsingular irreducible projective curves over $\F_q$, via a well-known correspondence that is functorial up to isomorphism \cite[Section~1.6]{har77}. In this paper, we adopt the algebraic language of function fields. For a comprehensive treatment, see \cite{Sti09}.

Throughout this subsection, let $F$ be a finitely generated function field of transcendence degree one over $\F_q$.

\paragraph{Discrete valuations.}
A \emph{discrete valuation} on $F$ (trivial on $\F_q$) is a map $v:F\to \mathbb{Z}\cup\{\infty\}$
such that for all $x,y\in F$:
\begin{enumerate}
\item $v(xy)=v(x)+v(y)$;
\item $v(x+y)\ge \min\{v(x),v(y)\}$;
\item $v(x)=\infty \iff x=0$;
\item $v(x)=0$ for all $x\in\F_q^\times$.\footnote{This condition is redundant: it follows from $v(xy)=v(x)+v(y)$ and the finiteness of $\F_q^\times$, but we state it explicitly for clarity.}
\end{enumerate}
We say $v$ is \emph{normalized} if $v(F)=\Z$.

When $v(x)\neq v(y)$, the second condition is in fact an equality:
\begin{lemma}\label{lem:triangle}
If $v(x)\neq v(y)$, then $v(x+y)=\min\{v(x),v(y)\}$.
\end{lemma}

\begin{proof}
By symmetry, we may assume $v(x)<v(y)$.
We have $v(x+y)\ge \min\{v(x),v(y)\}=v(x)$.
Also,
\begin{equation}\label{eq:STI}
v(x)=v((x+y)-y)\ge \min\{v(x+y),v(-y)\}=\min\{v(x+y),v(y)\}.
\end{equation}

Since $v(x)<v(y)$, the minimum on the right-hand side of \eqref{eq:STI} cannot be $v(y)$; hence it must be
$v(x+y)$. 
So \eqref{eq:STI} becomes $v(x)\ge v(x+y)$. Combining with $v(x+y)\ge v(x)$ yields
$v(x+y)=v(x)$.
\end{proof}

\paragraph{Valuation rings.}
Given a discrete valuation $v$ on $F$, the associated \emph{valuation ring} is
\[
\mathcal{O}_v=\{x\in F : v(x)\ge 0\}.
\]
It is a local ring, i.e., a ring with a unique maximal ideal.
This maximal ideal is
\[
\mathfrak{m}_v=\{x\in F : v(x)>0\}.
\]
An element $x\in \mathcal{O}_v$ is a unit if and only if $v(x)=0$.
The \emph{residue field} of $v$ is defined as $\kappa_v=\mathcal{O}_v/\mathfrak{m}_v$.


\paragraph{Places.}
A \emph{place} $P$ of $F$ consists of a normalized
discrete valuation
$v_P : F \to \mathbb{Z} \cup \{\infty\}$
that is trivial on $\mathbb{F}_q$. We denote by $\mathcal{P}_F$ the set of all places of $F$. For a place $P \in \mathcal{P}_F$, we write $\mathcal{O}_P$ for its valuation ring, $\mathfrak{m}_P$ for its maximal ideal, and $\kappa_P = \mathcal{O}_P / \mathfrak{m}_P$ for the corresponding residue field. The degree of $P$ is defined as $\deg(P) = [\kappa_P : \mathbb{F}_q]$. A place $P$ is called \emph{rational} if $\deg(P)=1$, or equivalently if
$\kappa_P=\mathbb{F}_q$. For such a place $P$, the \emph{evaluation} of a function
$f\in \mathcal{O}_P$ at $P$, denoted by $f(P)\in\mathbb{F}_q$, is defined as the image of $f$ in
the residue field $\kappa_P$.

\paragraph{Divisors.}
A \emph{divisor} on $F$ is a formal integer linear combination of places,
\[
D = \sum_{P \in \mathcal{P}_F} n_P P,
\]
where all but finitely many coefficients $n_P$ are zero.
The degree of $D$ is $\deg(D) = \sum_{P\in \mathcal{P}_F}  n_P \deg(P)$. We call $D$ \emph{effective}, and write
$D\ge 0$, if $n_P\ge 0$ for all $P\in\mathcal{P}_F$.

The set of divisors on $F$ forms a free abelian group $\mathrm{Div}(F)$ generated by the places of $F$.
The degree map $\deg(\cdot):\mathrm{Div}(F) \to \mathbb{Z}$ is a group homomorphism.

For a nonzero function $f \in F^\times$, the associated \emph{principal divisor} is
\[
(f)=\sum_{P \in \mathcal{P}_F} v_P(f)\, P,
\]
where $v_P$ is the normalized discrete valuation corresponding to the place $P$. The degree of a principal divisor $(f)$ is always zero. Intuitively, this means that the sum of the orders of the zeros of $f$ equals the sum of the orders of its poles.

\paragraph{Riemann--Roch spaces.}
For a divisor $D$, we define the associated Riemann--Roch space
\[
\mathcal{L}(D)=\{f\in F^\times\cup\{0\} : (f)+D \ge 0\}.
\]
It is a finite-dimensional $\mathbb{F}_q$-vector space, and we denote its
dimension by $\ell(D)=\dim_{\mathbb{F}_q}\mathcal{L}(D)$.

We need the following lemma, which is a function-field generalization of the fact
that a nonzero polynomial of degree $d$ cannot have more than $d$ zeros.

\begin{lemma}\label{lem:zero-section}
$\ell(D)=0$ if $\deg(D)<0$.
\end{lemma}
\begin{proof}
Assume to the contrary that $\ell(D)>0$, i.e., $\mathcal{L}(D)$ contains a nonzero function $f\in F^\times$. Then $(f)+D\geq 0$. So $\deg((f)+D)\geq 0$. On the other hand, $\deg((f)+D)=\deg((f))+\deg(D)=\deg(D)$ since the degree of the principal divisor $(f)$ is zero. But this is impossible since $\deg(D)<0$.
\end{proof}

The following important theorem gives an estimate on the dimension of a Riemann--Roch space.

\begin{theorem}[Riemann--Roch theorem]\label{thm:RR}
There exist an integer $g\ge 0$ and a divisor $K$ on $F$ such that for every divisor
$D$ on $F$,
\[
\ell(D)=\deg(D)+1-g+\ell(K-D).
\]
The integer $g$ is called the \emph{genus} of $F$, and $K$ is called a
\emph{canonical divisor}. Moreover, $\deg(K)=2g-2$.
\end{theorem}

As $\ell(K-D)=\dim_{\F_q}\mathcal{L}(K-D)\geq 0$, we have the following corollary.

\begin{corollary}[Riemann's inequality]\label{cor:riemann}
$\ell(D)\geq \deg(D)+1-g$.
\end{corollary}

This inequality is in fact an equality when the degree of $D$ is sufficiently large:

\begin{corollary}\label{cor:riemann-equality}
Suppose $\deg(D)\geq 2g-1$. Then $\ell(D)=\deg(D)+1-g$. 
\end{corollary}
\begin{proof}
We have $\deg(K-D)=\deg(K)-\deg(D)=(2g-2)-\deg(D)<0$. So $\ell(K-D)=0$ by \cref{lem:zero-section}. The claim then follows from the Riemann--Roch theorem.
\end{proof}

\cref{cor:riemann-equality} further yields the following corollary, which guarantees the existence of elements in a Riemann–Roch space with a prescribed order of zero or pole at a given place.

\begin{corollary}\label{cor:existence-of-functions}
Let $P$ be a place of $F$ and $D$ be a divisor on $F$. Let $n$ be the coefficient of $P$ in $D$. 
Suppose $\deg(D-P)\geq 2g-1$. Then there exists $f\in\mathcal{L}(D)$ such that $v_P(f)=-n$.
\end{corollary}

\begin{proof}
Applying \cref{cor:riemann-equality} to $D$ and $D-P$ shows that $\ell(D)=\deg(D)+1-g$ and $\ell(D-P)=\deg(D)+1-g-\deg(P)=\ell(D)-\deg(P)<\ell(D)$.
So there exists $f\in \mathcal{L}(D)$ that is not in $\mathcal{L}(D-P)$.
By definition, we have $v_P(f)+n\geq 0$ but $v_P(f)+(n-1)<0$. So $v_P(f)=-n$.
\end{proof}

We also record the following corollary of Riemann's inequality on the existence of functions with distinct valuations at a given rational place.

\begin{corollary}\label{cor:interval}
Let $P$ be a rational place of $F$. Let $n$ be a nonnegative integer and $d=n-1+g$. Then there exist $f_1,f_2,\dots,f_n\in \mathcal{L}(d P)$ such that
\[
0\geq v_P(f_1)>v_P(f_2)>\dots >v_P(f_n)\geq -d.
\]
\end{corollary}

\begin{proof}
For $t\ge 0$, let $V_t=\mathcal{L}((d-t)P)$. Then
\begin{equation}\label{eq:chain-V}
V_0 \supseteq V_1 \supseteq \cdots \supseteq V_{d+1}=\{0\},
\end{equation}
where the last equality holds by \cref{lem:zero-section}.
By Riemann's inequality (\cref{cor:riemann}), we have $\dim_{\F_q} V_0=\ell(dP)\ge d+1-g=n$.

Moreover, $\dim_{\F_q} V_t/V_{t+1}\le 1$ for every $t$. Indeed, let $u\in F^\times$ satisfy $v_P(u)=1$. Then the map
\[
V_t \to \kappa_P,\qquad f\mapsto (u^{d-t} f)(P),
\]
is well-defined, and its kernel is exactly $V_{t+1}$. Hence $V_t/V_{t+1}$ embeds into $\kappa_P\cong\F_q$, so it has dimension at most one over $\F_q$.

As $\dim_{\F_q} V_t$ drops from at least $n$ to $0$ along the chain \eqref{eq:chain-V}, and each step decreases the dimension by at most $1$, there are at least $n$ indices $0\le t_1<\cdots<t_n\le d$ such that $\dim V_{t_i}/V_{t_i+1}=1$.

For each $i\in[n]$, choose $f_i\in V_{t_i}\setminus V_{t_i+1}$. Then $f_i\in \mathcal{L}(dP)$ and $v_P(f_i)=-t_i$. Since $0\leq t_1<\cdots<t_n\leq d$, we obtain
\[
0\ge v_P(f_1)>v_P(f_2)>\cdots>v_P(f_n)\ge -d.
\]
This proves the claim.
\end{proof}

Finally, we need the following lemma, which lower-bounds the valuation of the determinant at a place by the corank of the reduced matrix.

\begin{lemma}\label{lem:val-from-rank}
Let $P$ be a place of $F$, and let $A\in \mathcal{O}_P^{r\times r}$. Let
$\overline{A}\in \kappa_P^{r\times r}$ denote the matrix obtained by reducing the
entries of $A$ modulo $\mathfrak{m}_P$. Then
\[
v_P(\det(A))\ge r-\rank(\overline{A}).
\]
\end{lemma}

\begin{proof}
Let $s:=\rank(\overline{A})$. Choose $s$ rows of $\overline{A}$ that form a basis
of its row space, and let $A_{i_1},\dots,A_{i_s}$ be the corresponding rows of
$A$. For every other row $A_i$, write
$\overline{A}_i=\sum_{j=1}^s c_{i,j}\overline{A}_{i_j}$ with $c_{i,j}\in\kappa_P$
and lift each $c_{i,j}$ to some $\widetilde c_{i,j}\in\mathcal O_P$.
Replacing $A_i$ by $A_i-\sum_{j=1}^s \widetilde c_{i,j}A_{i_j}$
does not change the determinant. Performing these row operations for all
$i\notin\{i_1,\dots,i_s\}$ yields a matrix $B\in\mathcal O_P^{r\times r}$
with $\det(B)=\det(A)$ such that all entries in $r-s$ rows of $B$ lie in
$\mathfrak m_P$.

In the determinant expansion of $B$, every term contains one entry from each
of these $r-s$ rows. Hence every term lies in $\mathfrak m_P^{r-s}$, and therefore $\det(B)\in\mathfrak m_P^{r-s}$.
Thus
\[
v_P(\det(A))
=
v_P(\det(B))
\ge r-s
=
r-\rank(\overline{A}).\qedhere
\]
\end{proof}

\subsection{Explicit Function Field Towers}

Our constructions use explicit function field towers $F_1\subseteq F_2\subseteq \cdots$ that are \emph{asymptotically good}, meaning that the ratio $N_i/g_i$ between the number $N_i$ of rational places of $F_i$ and its genus $g_i$ is bounded away from zero by a positive constant. Moreover, as $q$ varies, the towers we use satisfy $N_i/g_i = \Omega(q^c)$ for some constant $c>0$.

\paragraph{The Garcia--Stichtenoth tower.}

The Garcia--Stichtenoth tower constructed in \cite{GS96} is a well-known sequence
of function fields over finite fields with particularly favorable asymptotic properties.
An earlier asymptotically optimal tower was constructed in \cite{GS95}.
We briefly recall the definition and basic properties of the tower from \cite{GS96},
which will play a central role in our constructions. 

\begin{definition}[{Garcia--Stichtenoth tower \cite{GS96}}]
Let $q$ be a square, and write $q=\ell^2$ for some prime power $\ell$.
The Garcia--Stichtenoth tower over $\F_q$ is the sequence of function fields
\[
F_1 \subseteq F_2 \subseteq \cdots
\]
defined recursively as follows.
Let $F_1=\F_q(x_1)$ be the rational function field, and for $i\ge 1$, define
\[
F_{i+1}=F_i(x_{i+1}),
\qquad
x_{i+1}^\ell + x_{i+1} = \frac{x_i^\ell}{x_i^{\ell-1}+1}.
\]
\end{definition}

The following lemma gives estimates for the number of rational places and the genus of the fields in the Garcia--Stichtenoth tower.

\begin{lemma}[{\cite{GS96}}]\label{lem:GS}
For $i\geq 1$, let $N_i$ and $g_i$ denote the number of rational places and the genus of the $i$-th field $F_i$ in the Garcia--Stichtenoth tower over $\F_q$, where $q=\ell^2$. Then
\[
[F_i:F_1]=\ell^{i-1},
\]
\[
N_i\geq \ell^i(\ell-1)+1,
\]
and
\[
g_i =
\begin{cases}
(\ell^{i/2}-1)^2, & \text{if } i \text{ is even,}\\[4pt]
(\ell^{(i+1)/2}-1)(\ell^{(i-1)/2}-1), & \text{if } i \text{ is odd.}
\end{cases}
\]
In particular, $g_i\leq \ell^i$.
\end{lemma} 

In particular, we have
\[
\limsup_{i\to\infty}\frac{N_i}{g_i}\ge \ell-1=\sqrt{q}-1,
\]
which matches the largest possible asymptotic ratio between the number of rational places and the genus of a function field over $\F_q$, known as the Drinfeld--Vl\u{a}du\c{t} bound~\cite{VS83}.

\paragraph{The Bassa--Beelen--Garcia--Stichtenoth tower.}
The Garcia--Stichtenoth tower is defined over $\F_q$ only when $q$ is a square. When $q=p^{2m+1}$ with $m\ge 1$, we instead use the Bassa--Beelen--Garcia--Stichtenoth tower~\cite{BBGS15}. We recall below a special case of this construction.

\begin{definition}[{Bassa--Beelen--Garcia--Stichtenoth tower \cite{BBGS15}}]
Let $q=p^{2m+1}$ be a prime power with $m\geq 1$, where $p$ is a prime.
Define
\[
j=\begin{cases}
m & \text{if } p \nmid m,\\
m+1 & \text{otherwise,}
\end{cases}
\qquad
k=2m+1-j\in\{m,m+1\}.
\]
For an integer $a\geq 1$, define
\[
\Tr_a(T):=T+T^p+\cdots+T^{p^{a-1}}\in\F_q[T].
\]
The Bassa--Beelen--Garcia--Stichtenoth tower over $\F_q$ is the sequence of function fields
\[
E_1 \subseteq E_2 \subseteq \cdots
\]
defined recursively as follows.
Let $E_1=\F_q(x_1)$ be the rational function field, and for $i\ge 1$, define
\[
E_{i+1}=E_i(x_{i+1}),
\qquad
\Tr_j\!\left(\frac{x_{i+1}}{x_i^{p^k}}\right)
+
\Tr_k\!\left(\frac{x_{i+1}^{p^j}}{x_i}\right)
=1.
\]
\end{definition}

The following lemma gives estimates for the number of rational places and the genus of the fields in the Bassa--Beelen--Garcia--Stichtenoth tower.

\begin{lemma}[{\cite{BBGS15}}]\label{lem:BGGS}
Let $q=p^{2m+1}$ be a prime power with $m\geq 1$, where $p$ is a prime. For $i\ge 1$, let $N_i$ and $g_i$ denote the number of rational places and the genus of the $i$-th field $E_i$ in the Bassa--Beelen--Garcia--Stichtenoth tower over $\F_q$. Then
\[
[E_i:E_1]=p^{2m(i-1)},
\]
\[
N_i\ge [E_i:E_1]\,(q-1)=p^{2m(i-1)}(q-1),
\]
and
\[
g_i \le \frac{[E_i:E_1]}{2}\left(\frac{q-1}{p^{m}-1}+\frac{q-1}{p^{m+1}-1}\right)= \frac{p^{2m(i-1)}}{2}\left(\frac{q-1}{p^{m}-1}+\frac{q-1}{p^{m+1}-1}\right).
\]
\end{lemma}

In particular, we have
\[
\limsup_{i\to\infty}\frac{N_i}{g_i}\ge
2\left(\frac{1}{p^{m}-1}+\frac{1}{p^{m+1}-1}\right)^{-1}.
\]

This does not match the Drinfeld--Vl\u{a}du\c{t} bound~\cite{VS83},
which in this case is $\sqrt{q}-1=p^{m+1/2}-1$, but it is polynomially related to it.

\begin{remark}
Our constructions based on function fields require the ability to compute bases of
Riemann--Roch spaces $\mathcal{L}(dP_\infty)$ associated with a distinguished rational place $P_\infty$, as well as to evaluate such functions at rational places.

For the function field towers we use, namely the Garcia--Stichtenoth and
Bassa--Beelen--Garcia--Stichtenoth towers, these operations can be carried out
in polynomial time using standard algorithms from computational algebraic function field
theory; see, e.g., \cite{He02,SAKSD02}.
\end{remark}

\section{A Function-Field Analog of the Gabizon--Raz Construction}\label{sec:basic}

In this section, we present a construction that can be viewed as a function-field analog of the construction of Gabizon and Raz~\cite{GR08}.

\paragraph{Construction.}

Let $1\le r\le k$ be integers.
Let $F$ be a finitely generated function field of transcendence degree one over a finite field $\F_q$, and let $g$ be its genus.
Let $P_\infty$ be the distinguished rational place of $F$, and let $v_{P_\infty}$ be its normalized valuation. Let $S$ be the set of rational places of $F$ other than $P_\infty$.

For each $i\in [r]$ and $j\in [k]$, define
\[
d_{ij}=i\cdot (j-1)+2g\geq 2g,
\]
and choose $f_{ij}\in \mathcal{L}(d_{ij}P_\infty)$ such that $v_{P_\infty}(f_{ij})=-d_{ij}$.
The existence of such a function follows from \cref{cor:existence-of-functions} applied to $D=d_{ij}P_\infty$, since
$\deg(d_{ij}P_\infty-P_\infty)=d_{ij}-1\ge 2g-1$.

Construct the $r\times k$ matrix $E$ over $F$ by
\[
E=(f_{ij})_{i\in [r],\, j\in [k]}.
\]

For each rational place $P\in S$, define the $r\times k$ matrix $E(P)$ over $\F_q$ by
\[
E(P)=(f_{ij}(P))_{i\in [r],\, j\in [k]}.
\]

\paragraph{Analysis.}

The following lemma establishes a nonvanishing property of certain determinants.

\begin{lemma}\label{lem:GR-deg}
Let $M\in \mathbb{F}_q^{k\times r}$ be a full-rank matrix. Then $\det(EM)\in F$ is nonzero, $v_P(\det(EM))\geq 0$ for all places $P\neq P_\infty$, and
\[
v_{P_\infty}(\det(EM))\geq -r\left(\frac{(r+1)(3k-r-2)}{6}+2g\right).
\]
\end{lemma}

\begin{proof}
For any place $P\neq P_\infty$, we have $v_P(\det(EM))\geq 0$ since $v_P(a)\geq 0$ whenever $a$ is an entry of $E$ or an entry of $M$.

By the Cauchy--Binet formula (\cref{lem:CB}),
\[
\det(EM)=\sum_{I\subseteq [k],\, |I|=r}\det(E_I)\det(M^I).
\]

Fix $I=\{j_1<\cdots<j_r\}\subseteq [k]$. Expanding $\det(E_I)$, we write
\[
\det(E_I)=\sum_{\sigma\in \mathrm{Sym}(r)} T_\sigma,
\quad
T_\sigma:=\mathrm{sgn}(\sigma)\prod_{i=1}^r f_{i j_{\sigma(i)}}.
\]

We claim that the identity permutation uniquely minimizes $v_{P_\infty}(T_\sigma)$. Indeed,
\[
v_{P_\infty}(T_\sigma) = \sum_{i=1}^r v_{P_\infty}(f_{ij_{\sigma(i)}}) = -\sum_{i=1}^r d_{ij_{\sigma(i)}}
= -\sum_{i=1}^r \bigl(i\cdot (j_{\sigma(i)}-1)+2g\bigr).
\]
If $\sigma$ has an inversion $(i,i')$, i.e., $i<i'$ but $\sigma(i)>\sigma(i')$, then
$j_{\sigma(i)}>j_{\sigma(i')}$.
Let $\sigma'$ be obtained from $\sigma$ by swapping $\sigma(i)$ and $\sigma(i')$.
Then
\[
v_{P_\infty}(T_{\sigma'}) - v_{P_\infty}(T_\sigma)
= -\bigl(i\,j_{\sigma(i')} + i'j_{\sigma(i)}\bigr)
+ \bigl(i\,j_{\sigma(i)} + i'j_{\sigma(i')}\bigr)
= (i'-i)\bigl(j_{\sigma(i')}-j_{\sigma(i)}\bigr) < 0,
\]
since $i'<i$ is false and $j_{\sigma(i')}-j_{\sigma(i)}<0$.
Thus swapping strictly decreases the valuation. Hence the identity is the unique minimizer.

By \cref{lem:triangle}, it follows that
\[
v_{P_\infty}(\det(E_I))
= -\sum_{i=1}^r \bigl(i\cdot (j_i-1)+2g\bigr)
= -\sum_{i=1}^r i\,j_i - \sum_{i=1}^r (2g-i).
\]

Since $\det(M^I)\in \F_q$, we have $v_{P_\infty}(\det(M^I))\in \{0,\infty\}$, and at least one $I$ satisfies $\det(M^I)\neq 0$ because $M$ has full rank.

Thus the minimum of $v_{P_\infty}(\det(E_I)\det(M^I))$ is attained uniquely by the subset $I^*$ minimizing $-\sum_{i=1}^r i\,j_i$ subject to $\det(M^I)\neq 0$. By \cref{lem:unique}, such $I^*$ is unique.

Let $I^*=\{j_1^*<\cdots<j_r^*\}$. Since the minimum valuation in the above expansion is attained uniquely at $I^*$, \cref{lem:triangle} implies that
\[
v_{P_\infty}(\det(EM))
= -\sum_{i=1}^r i\,j_i^* - \sum_{i=1}^r (2g-i)
< \infty.
\]
In particular, $\det(EM)\neq 0$.

Moreover,
\[
v_{P_\infty}(\det(EM))
\ge -\sum_{i=1}^r i(k-r+i) - \sum_{i=1}^r (2g-i),
\]
which simplifies to
\[
v_{P_\infty}(\det(EM))
\ge -r\left(\frac{(r+1)(3k-r-2)}{6}+2g\right).\qedhere
\]
\end{proof}

Now we prove the main result of this section.

\begin{theorem}\label{thm:bound-bad-ones}
Let $M\in \mathbb{F}_q^{k\times r}$ be a full-rank matrix.
Then 
\[
\sum_{P\in S} (r-\rank(E(P)M))\leq r\left(\frac{(r+1)(3k-r-2)}{6}+2g\right).
\]
\end{theorem}

\begin{proof}
Let
\[
d=r\left(\frac{(r+1)(3k-r-2)}{6}+2g\right).
\]
By \cref{lem:GR-deg}, we have $\det(EM)\neq 0$ and $\det(EM)\in \mathcal{L}(d P_\infty)$.

Consider $P\in S$. As the entries of $E$ are in $\mathcal{O}_P$, we have $E M\in\mathcal{O}_P^{r\times r}$. 
By \cref{lem:val-from-rank}, we have $v(\det(EM))\geq r-\rank(E(P)M)$. Therefore, taking all $P\in S$ into account, we have
\[
\det(EM)\in \mathcal{L}(D) \quad\text{where}\quad D:=dP_\infty-\sum_{P\in S} (r-\rank(E(P)M)) P.
\]
Since $\det(EM)\neq 0$, \cref{lem:zero-section} implies that $\deg(D)\geq 0$, i.e., $\sum_{P\in S} (r-\rank(E(P)M))\le d$.
\end{proof}

One can instantiate the function field and the parameters to obtain explicit lossless rank extractors over fields of size at least $\poly(r)$. We defer this to the next section, where we present an alternative construction that improves the $O(r^2 k)$ term in \cref{thm:bound-bad-ones} to $O(rk)$.

\section{A Function-Field Analog of the Forbes--Shpilka Construction}\label{sec:FS}

We now present an alternative construction, which can be viewed as a function-field analog of the construction of Forbes and Shpilka~\cite{FS12}.

\paragraph{Construction.}

Let $1\le r\le k$ be integers.
Let $F$ be a finitely generated function field of transcendence degree one over a finite field $\F_q$, and let $g$ be its genus.
Let $P_\infty$ be the distinguished rational place of $F$, and let $v_{P_\infty}$ be its normalized valuation. Let $S$ be the set of rational places of $F$ other than $P_\infty$. 

Let $h=\left\lceil\frac{k}{q-1}\right\rceil$.
Let $d_1=k-1+g$ and $d_2=h-1+g$. By Riemann's inequality (\cref{cor:riemann}), we have $\ell(d_1P_\infty)\geq k$ and $\ell(d_2P_\infty)\geq h$.

By \cref{cor:interval}, there exist $f_1,\dots,f_k\in \mathcal{L}(d_1P_\infty)$ and $g_1,\dots,g_h\in \mathcal{L}(d_2 P_\infty)$ such that
\begin{equation}\label{eq:dec-ord-1}
0\geq v_{P_\infty}(f_1) > \cdots > v_{P_\infty}(f_k) \geq -d_1
\end{equation}
and
\begin{equation}\label{eq:dec-ord-2}
0\geq v_{P_\infty}(g_1) > \cdots > v_{P_\infty}(g_h) \geq -d_2.
\end{equation}

For $j\in [k]$, define $\alpha(j)=\lfloor\frac{j-1}{q-1}\rfloor+1\in\{1,\dots,h\}$ and $\beta(j)=(j-1)\bmod (q-1)\in \{0,\dots,q-2\}$.
Let $\gamma$ be a generator of $\F_q^\times$.

Construct the $r\times k$ matrix $E$ over $F$ by
\[
E_{ij}=(\gamma^{\beta(j)} g_{\alpha(j)})^{i-1}  f_j.
\]

For each $P\in S$, define
\[
E(P)=(E_{ij}(P))_{i\in [r], j\in [k]}.
\]

\begin{remark}
Instead of using \cref{cor:interval}, which follows from Riemann's inequality (\cref{cor:riemann}), one can use \cref{cor:existence-of-functions}, as in \cref{sec:basic}, which relies on the full Riemann--Roch theorem rather than its weaker consequence. This yields functions $f_i$ and $g_i$ with prescribed valuations $v_{P_\infty}(f_i)$ and $v_{P_\infty}(g_i)$, rather than only the bounds in \eqref{eq:dec-ord-1} and \eqref{eq:dec-ord-2}. This increases the $g$ term in the analysis to $2g$, but does not affect the asymptotics up to constant factors.
\end{remark}

\paragraph{Analysis.}

Next, we prove the following analog of \cref{lem:GR-deg} concerning the nonvanishing of determinants.

\begin{lemma}\label{lem:FS-deg}
Let $M\in \mathbb{F}_q^{k\times r}$ be full rank. Then $\det(EM)\in F$ is nonzero, $v_P(\det(EM))\geq 0$ for all places $P\neq P_\infty$, and 
\[
v_{P_\infty}(\det(EM))\geq -r\left(
k-r+\frac{r-1}{2}\left\lceil \frac{k}{q-1}\right\rceil+\frac{r+1}{2}g 
\right).
\]
\end{lemma}
\begin{proof}
For any place $P\neq P_\infty$, we have $v_P(\det(EM))\geq 0$ since $v_P(a)\geq 0$ whenever $a$ is an entry of $E$ or an entry of $M$.

By the Cauchy--Binet formula,
\[
\det(EM)=\sum_{I\subseteq [k],\, |I|=r}\det(E_I)\det(M^I).
\]

Fix $I=\{j_1<\cdots<j_r\}\subseteq [k]$. The matrix $E_I$ has the form
\[
E_I=\bigl((\gamma^{\beta(j_t)} g_{\alpha(j_t)})^{i-1} f_{j_t}\bigr)_{i,t\in [r]}.
\]
Factoring out $f_{j_t}$ from each column and using multilinearity of the determinant, we obtain
\[
\det(E_I)=\left(\prod_{u=1}^r f_{j_u}\right)
\det\bigl((\gamma^{\beta(j_t)} g_{\alpha(j_t)})^{i-1}\bigr)_{i,t\in [r]}.
\]
The remaining determinant is a Vandermonde determinant, and hence
\begin{equation}\label{eq:expansion-EI}
\det(E_I)=\left(\prod_{u=1}^r f_{j_u}\right)
\prod_{1\le t_1<t_2\le r}
\left(\gamma^{\beta(j_{t_2})} g_{\alpha(j_{t_2})}-\gamma^{\beta(j_{t_1})} g_{\alpha(j_{t_1})}\right).
\end{equation}

\medskip
We now analyze the valuation of each difference term. For $t_1<t_2$, consider
\[
\gamma^{\beta(j_{t_2})} g_{\alpha(j_{t_2})}
-
\gamma^{\beta(j_{t_1})} g_{\alpha(j_{t_1})}.
\]
There are two cases.

\smallskip
\noindent\textbf{Case 1:} $\alpha(j_{t_1})=\alpha(j_{t_2})$.
Then $\beta(j_{t_1})<\beta(j_{t_2})$, and
\[
v_{P_\infty}\!\left(
\gamma^{\beta(j_{t_2})} g_{\alpha(j_{t_2})}
-
\gamma^{\beta(j_{t_1})} g_{\alpha(j_{t_1})}
\right)
=
v_{P_\infty}\!\left(
(\gamma^{\beta(j_{t_2})}-\gamma^{\beta(j_{t_1})}) g_{\alpha(j_{t_2})}
\right)
=
v_{P_\infty}(g_{\alpha(j_{t_2})}).
\]

\smallskip
\noindent\textbf{Case 2:} $\alpha(j_{t_1})<\alpha(j_{t_2})$.
Then by \eqref{eq:dec-ord-2},
\[
v_{P_\infty}(g_{\alpha(j_{t_2})})<v_{P_\infty}(g_{\alpha(j_{t_1})}),
\]
and hence, by \cref{lem:triangle},  $v_{P_\infty}\!\left(
\gamma^{\beta(j_{t_2})} g_{\alpha(j_{t_2})}
-
\gamma^{\beta(j_{t_1})} g_{\alpha(j_{t_1})}
\right)
=
v_{P_\infty}(g_{\alpha(j_{t_2})})$.

\smallskip
Thus in both cases, 
\begin{equation}\label{eq:diff}
v_{P_\infty}\!\left(
\gamma^{\beta(j_{t_2})} g_{\alpha(j_{t_2})}
-
\gamma^{\beta(j_{t_1})} g_{\alpha(j_{t_1})}
\right)
=
v_{P_\infty}(g_{\alpha(j_{t_2})}).
\end{equation}

\medskip
For each $I=\{j_1<\cdots<j_r\}$, define
\[
T_I=\det(E_I)\det(M^I).
\]
From \eqref{eq:expansion-EI} and the definition of $T_I$, we obtain
\[
v_{P_\infty}(T_I)
=
\sum_{u=1}^r v_{P_\infty}(f_{j_u})
+
\sum_{1\le t_1<t_2\le r}
v_{P_\infty}\!\left(
\gamma^{\beta(j_{t_2})} g_{\alpha(j_{t_2})}
-
\gamma^{\beta(j_{t_1})} g_{\alpha(j_{t_1})}
\right)
+
v_{P_\infty}(\det(M^I)).
\]
By \eqref{eq:diff}, each term in the second sum equals
$v_{P_\infty}(g_{\alpha(j_{t_2})})$, and hence
\begin{equation}\label{eq:expansion-TI}
\begin{aligned}
v_{P_\infty}(T_I)
&=
\sum_{u=1}^r v_{P_\infty}(f_{j_u})
+
\sum_{t=1}^r (t-1)v_{P_\infty}(g_{\alpha(j_t)})
+
v_{P_\infty}(\det(M^I)).\\
&=\begin{cases}
\sum_{u=1}^r v_{P_\infty}(f_{j_u})
+
\sum_{t=1}^r (t-1)v_{P_\infty}(g_{\alpha(j_t)}) & \text{if } \det(M^I)\neq 0,\\
\infty & \text{otherwise.}
\end{cases}
\end{aligned}
\end{equation}

Let $I^*\subseteq [k]$ be the subset produced by the following greedy procedure:
initialize $I^*=\emptyset$, and for $i=k,k-1,\dots,1$, add $i$ to $I^*$ if the $i$-th column of $M$ does not lie in the $\F$-linear span of the columns indexed by the current set $I^*$.
By \cref{lem:unique}, we have $|I^*|=r$ and $\det(M^{I^*})\neq 0$.

First, we show that $I^*$ uniquely minimizes
the first sum $\sum_{u=1}^r v_{P_\infty}(f_{j_u})$ in \eqref{eq:expansion-TI}
among all subsets $I\subseteq [k]$ of size $r$ satisfying $\det(M^I)\neq 0$.
Indeed, by \eqref{eq:dec-ord-1}, the sequence
\[
v_{P_\infty}(f_1) > v_{P_\infty}(f_2) > \cdots > v_{P_\infty}(f_k)
\]
is strictly decreasing. Applying the second item of \cref{lem:unique} with weights $w_i = v_{P_\infty}(f_i)$, we conclude that
\[
\sum_{i\in I^*} v_{P_\infty}(f_i)
\le \sum_{i\in J} v_{P_\infty}(f_i)
\]
for every subset $J\subseteq [k]$ of size $r$ such that $\det(M^J)\neq 0$, with equality if and only if $J=I^*$. This proves the claim.

Next, we show that $I^*$ also minimizes the second sum $\sum_{t=1}^r (t-1)v_{P_\infty}(g_{\alpha(j_t)})$
among all subsets $I=\{j_1<\cdots<j_r\}$ satisfying $\det(M^I)\neq 0$.

Let $I=\{j_1<\cdots<j_r\}$ be any such subset, and write $I^*=\{j_1^*<\cdots<j_r^*\}$.
By the first item of \cref{lem:unique}, we have $j_t^*\ge j_t$ for all $t\in [r]$.
Since $\alpha(\cdot)$ is nondecreasing and the sequence
$v_{P_\infty}(g_1) > \cdots > v_{P_\infty}(g_h)$
is strictly decreasing by \eqref{eq:dec-ord-2}, it follows that
\[
v_{P_\infty}(g_{\alpha(j_t^*)})
\le
v_{P_\infty}(g_{\alpha(j_t)})
\qquad\text{for all } t\in [r].
\]
Multiplying by $(t-1)\ge 0$ and summing over $t$, we obtain
\[
\sum_{t=1}^r (t-1)v_{P_\infty}(g_{\alpha(j_t^*)})
\le
\sum_{t=1}^r (t-1)v_{P_\infty}(g_{\alpha(j_t)}),
\]
as claimed.

Combining the two parts above, we conclude that $I^*$ uniquely minimizes $v_{P_\infty}(T_I)$ among all subsets $I\subseteq [k]$ of size $r$.

Returning to the Cauchy--Binet expansion
\[
\det(EM)=\sum_{I\subseteq [k],\, |I|=r} T_I,
\]
we have $\det(M^{I^*})\neq 0$, and hence \eqref{eq:expansion-TI} implies that
$v_{P_\infty}(T_{I^*}) < \infty$.
Since $T_{I^*}$ is the unique term attaining the minimum valuation, by \cref{lem:triangle},
\[
v_{P_\infty}(\det(EM)) = v_{P_\infty}(T_{I^*}) < \infty,
\]
and thus $\det(EM)\neq 0$.

Moreover, writing $I^*=\{j_1^*<\cdots<j_r^*\}$ and using \eqref{eq:expansion-TI}, we have
\begin{align*}
v_{P_\infty}(\det(EM))
&=
\sum_{u=1}^r v_{P_\infty}(f_{j_u^*})
+
\sum_{t=1}^r (t-1)v_{P_\infty}(g_{\alpha(j_t^*)})\\
&\stackrel{\eqref{eq:dec-ord-1},\eqref{eq:dec-ord-2}}{\ge}
-\sum_{u=1}^r(d_1-r+u)-\sum_{t=1}^r (t-1)d_2\\
&=-\sum_{u=1}^r(k-1+g-r+u)-\sum_{t=1}^r (t-1)\left(\left\lceil\frac{k}{q-1}\right\rceil-1+g\right),
\end{align*}
which simplifies to
\[
v_{P_\infty}(\det(EM))\geq -r\left(
k-r+\frac{r-1}{2}\left\lceil \frac{k}{q-1}\right\rceil+\frac{r+1}{2}g 
\right).\qedhere
\]
\end{proof}

The proof of the following theorem is identical to that of \cref{thm:bound-bad-ones}, except that we use \cref{lem:FS-deg} in place of \cref{lem:GR-deg}, improving the $O(r^2 k)$ term in \cref{thm:bound-bad-ones} to $O(rk)$.

\begin{theorem}\label{thm:bound-bad-ones-2}
Let $M\in \mathbb{F}_q^{k\times r}$ be a full-rank matrix. 
Then 
\[
\sum_{P\in S} (r-\rank(E(P)M))\leq r\left(
k-r+\frac{r-1}{2}\left\lceil \frac{k}{q-1}\right\rceil+\frac{r+1}{2}g
\right).
\]
\end{theorem}

\begin{remark}\label{rem:fs}
If we take $F=\F_q(X)$ with $g=0$ and $q>k$, then the above construction reduces to (a generalization of) that of Forbes and Shpilka~\cite{FS12}. In this case, the bound in \cref{thm:bound-bad-ones-2} becomes
\[
r\left(k-\frac{r+1}{2}\right)=rk-\binom{r+1}{2},
\]
recovering the bound in \cite{FS12}, except that \cite{FS12} established only the weak lossless rank extractor property.

For the construction of Forbes and Shpilka \cite{FS12}, this bound was later improved to $r(k-r)$ in~\cite{FSS14,For14}, using additional properties of the monomials $X^{i-1}$ for $i\in [k]$. It is unclear whether a similar improvement can be obtained for general function fields, or for those defined over small fields that we consider here. In any case, such an improvement would be minor in our setting where $r\ll k$, and would be dominated by the error terms.
\end{remark}

\paragraph{Instantiating the function field and parameters.}

We now choose $F$ to be a function field from either the Garcia--Stichtenoth tower or the Bassa--Beelen--Garcia--Stichtenoth tower, in order to obtain explicit rank extractors over non-prime fields. This construction also yields explicit subspace designs; we defer the corresponding discussion to \cref{sec:combine}.

\begin{theorem}[Formal version of \cref{thm:rankboundinformal}]\label{thm:fullrankbound}
Let $1\le r<k$ be integers. 
Suppose $q$ is a prime power that is not a prime and satisfies $q\ge (2r)^{c_0}$ for a sufficiently large absolute constant $c_0$. Then there exists an explicit $(k,r,L)$ strong lossless rank extractor over $\F_q$ of size $n$ consisting of full-rank matrices only, where $L=O(r(k-r)q)$ and $L/n\le q^{-1/4}$.
Moreover, for each fixed $r$ and $q$, we have $L\le c_1 r(k-r)$ for some absolute constant $c_1>0$ and infinitely many $k$.
\end{theorem}

\begin{proof}
We will assume $r\leq k/2$, so that $k=O(k-r)$. If $r>k/2$, we can reduce to the $r\leq k/2$ case by applying \cref{lem:rc2design} to exchange the dimension and the codimension.

If $q=\ell^2$ is a square, let $i$ be the smallest positive integer such that $\ell^i\ge k/r$, and let $F=F_i$ be the $i$-th field in the Garcia--Stichtenoth tower over $\F_q$.
Otherwise, $q=p^{2m+1}$ with $m\ge 1$, and we let $i$ be the smallest positive integer such that
\[
\frac{p^{2m(i-1)}}{2}\left(\frac{q-1}{p^{m}-1}+\frac{q-1}{p^{m+1}-1}\right)\ge k/r,
\]
and let $F=F_i$ be the $i$-th field in the Bassa--Beelen--Garcia--Stichtenoth tower over $\F_q$.
In either case, this choice of $i$ ensures that $g\le (k/r)q$ by \cref{lem:GS,lem:BGGS}.\footnote{We elaborate on the case $q=p^{2m+1}$; the case $q=\ell^2$ is similar. If $i=1$, then $F=E_1=\F_q(x_1)$ has genus $g=0$, so $g\leq (k/r)q$ trivially. Suppose now that $i\geq 2$. By the minimality of $i$, we have $\frac{p^{2m(i-2)}}{2}\left(\frac{q-1}{p^{m}-1}+\frac{q-1}{p^{m+1}-1}\right)<k/r$. So $g\leq \frac{p^{2m(i-1)}}{2}\left(\frac{q-1}{p^{m}-1}+\frac{q-1}{p^{m+1}-1}\right)\leq q\cdot \frac{p^{2m(i-2)}}{2}\left(\frac{q-1}{p^{m}-1}+\frac{q-1}{p^{m+1}-1}\right)\leq (k/r)q$.}
(In fact, this extra factor of $q$ can be improved to $\ell$ or $p^{2m}$, but we use this crude bound for simplicity.)

Let
\[
L=r\left(
k-r+\frac{r-1}{2}\left\lceil \frac{k}{q-1}\right\rceil+\frac{r+1}{2}g
\right).
\]
By \cref{thm:bound-bad-ones-2}, the collection $\mathcal{E}=\{E(P): P\in S\}$ is an explicit $(k,r,L)$ strong lossless rank extractor over $\F_q$.

Since $g\le (k/r)q$ and $q\ge (2r)^{c_0}$, we obtain $L=O(rkq)=O(r(k-r)q)$.

Moreover, for infinitely many $k$, we have $g=O(k/r)$. For instance, in the case $q=\ell^2$, we may fix $i$ and take $k=\ell^i r$. In the case $q=p^{2m+1}$, we may take $i$ sufficiently large and set
\[
k=\left\lfloor r\cdot \frac{p^{2m(i-1)}}{2}\left(\frac{q-1}{p^{m}-1}+\frac{q-1}{p^{m+1}-1}\right)\right\rfloor.
\]
For such $k$, we obtain $L=O(rk)=O(r(k-r))$.

By \cref{lem:reduced}, we may assume $\mathcal{E}$ consists of full-rank matrices only by decreasing the size $n$ by at most $L$ if necessary. So $n\geq |S|-L$.

Finally, we bound $L/n$. If $q=\ell^2$, then $|S|\ge \ell^i(\ell-1)$ by \cref{lem:GS}. By the choice of $i$, we have $L=O(r^2\ell^i)$, so
\[
L/|S|=O(r^2\ell^i/|S|)=O(r^2q^{-1/2}),
\]
which is at most $\frac{1}{2}q^{-1/4}$ for sufficiently large $c_0$.
Then $n\geq |S|-L\geq \frac{|S|}{2}$, implying that $L/n\leq q^{-1/4}$.

If $q=p^{2m+1}$, then $|S|\ge p^{2m(i-1)}(q-1)-1$ by \cref{lem:BGGS}. By the choice of $i$, we have
\[
L=O\!\left(r^2 p^{2m(i-1)}\frac{q-1}{p^m}\right),
\]
and hence
\[
L/|S|=O(r^2/p^m)=O\!\left(r^2 q^{-m/(2m+1)}\right),
\]
which is again at most $q^{-1/4}$ for sufficiently large $c_0$. Then again, we have $n\geq |S|-L\geq \frac{|S|}{2}$, implying that $L/n\leq q^{-1/4}$.
\end{proof}

\section{Field Reduction}\label{sec:concatenation}

As mentioned in the introduction, for a field extension $\F_Q/\F_q$, lossless rank condensers over $\F_Q$ can be converted into ones over $\F_q$ via \cite[Proposition~8.5]{FG15}, but this transformation does not preserve the lossless rank extractor or disperser property. In \cref{sec:det}, we develop a different field-reduction technique based on the multilinearity of the determinant, and show that it transforms lossless rank dispersers over $\F_Q$ into ones over $\F_q$, albeit with an exponential blow-up in size in $r$.

In \cref{sec:pit}, we further refine this approach by combining it with polynomial identity testing. This is particularly useful over sufficiently large prime fields $\F_q$, where techniques such as the Garcia--Stichtenoth tower and the Bassa--Beelen--Garcia--Stichtenoth tower do not directly apply.

Finally, in \cref{sec:combine}, we apply these techniques to prove our main theorems.

\subsection{Field Reduction via Multilinearity of the Determinant}\label{sec:det}

Let $r\leq k$ be positive integers.
Let $\F_Q/\F_q$ be a finite field extension of degree $d=[\F_Q:\F_q]$.

Fix a basis $e_1,\dots,e_d\in \F_Q$ of $\F_Q$ over $\F_q$.
For $a\in\F_Q$, let $a^{(1)},\dots,a^{(d)}\in\F_q$ denote the coordinates of $a$ with respect to this basis, so that $a=\sum_{t=1}^d a^{(t)} e_t$.

For an $r\times k$ matrix $E=(E_{ij})\in \F_Q^{r\times k}$, let $E_i$ denote its $i$-th row for $i\in [r]$. 
For $t\in [d]$, define
\[
E_i^{(t)} := (E_{i1}^{(t)},\dots,E_{ik}^{(t)})\in \F_q^k,
\]
the vector obtained by taking the $t$-th coordinate of each entry of $E_i$.

For a tuple $\sigma=(\sigma_1,\dots,\sigma_r)\in [d]^r$, define the matrix $E^{(\sigma)}\in \F_q^{r\times k}$ by
\[
E^{(\sigma)} := (E_i^{(\sigma_i)})_{i\in [r]}.
\]

Finally, for a finite collection $\mathcal{E}\subseteq\F_Q^{r\times k}$, define the collection
\[
\mathcal{E}_{\F_Q\to\F_q}
:= \{E^{(\sigma)} : E\in \mathcal{E},\ \sigma\in [d]^r\}
\subseteq \F_q^{r\times k}.
\]
Then
\[
|\mathcal{E}_{\F_Q\to\F_q}|= d^r |\mathcal{E}| = (\log_q Q)^r |\mathcal{E}|.
\]

Recall the definition of a lossless rank disperser (\cref{defi:disperser}). 
We now show that if $\mathcal{E}$ is a lossless rank disperser over $\F_Q$, then $\mathcal{E}_{\F_Q\to\F_q}$ is a lossless rank disperser over $\F_q$.
 
\begin{theorem} \label{thm:ReverseExtension}
Let $\mathcal{E}\subseteq\F_Q^{r\times k}$ be a finite collection.
Suppose $M\in \F_Q^{k\times r}$ is a full-rank matrix such that
\[
|\{E\in \mathcal{E}:\rank(EM)=r\}|>0.
\]
Then
\[
|\{E\in \mathcal{E}_{\F_Q\to\F_q}:\rank(EM)=r\}|>0.
\]
In particular, this holds for every $M\in \F_q^{k\times r}$ of rank $r$.
\end{theorem}

\begin{corollary}\label{cor:field-reduction}
If $\mathcal{E}$ is a $(k,r)$ lossless rank disperser over $\F_Q$, then $\mathcal{E}_{\F_Q\to\F_q}$ is a $(k,r)$ lossless rank disperser over $\F_q$.
\end{corollary}

\begin{proof}[Proof of \cref{thm:ReverseExtension}]
By assumption, there exists $E\in\mathcal{E}$ such that $\rank(EM)=r$, or equivalently, $\det(EM)\neq 0$. Fix such an $E=(E_{ij})$.

For each $i\in [r]$ and $j\in [k]$, we have
\[
E_i = \sum_{t\in [d]} e_{t} E_i^{(t)},
\]
and hence
\begin{equation}\label{eq:row-expansion}
E_i M = \sum_{t\in [d]} e_{t} \left(E_i^{(t)} M\right).
\end{equation}

For $\sigma=(\sigma_1,\dots,\sigma_r)\in [d]^r$, let $A_\sigma$ be the $r\times r$ matrix whose $i$-th row is $E_i^{(\sigma_i)} M$ for $i\in [r]$. Then $A_\sigma = E^{(\sigma)} M$ since the $i$-th row of $E^{(\sigma)} M$ is also $E_i^{(\sigma_i)} M$ for $i\in [r]$.

By \eqref{eq:row-expansion} and multilinearity of the determinant with respect to rows, we have
\[
\det(EM)
= \sum_{\sigma\in [d]^r}
\left(\prod_{i=1}^r e_{\sigma_i}\right)\det(A_\sigma)
= \sum_{\sigma\in [d]^r}
\left(\prod_{i=1}^r e_{\sigma_i}\right)\det(E^{(\sigma)} M).
\]

Since $\det(EM)\neq 0$, there exists $\sigma\in [d]^r$ such that $\det(E^{(\sigma)} M)\neq 0$. As $E^{(\sigma)}\in \mathcal{E}_{\F_Q\to\F_q}$, the claim follows.
\end{proof}

\subsection{Field Reduction via PIT for Symbolic Determinants with Rank-One Summands}\label{sec:pit}

\cref{cor:field-reduction} allows us to construct lossless rank dispersers over an arbitrary finite field $\F_q$ from those over an extension field $\F_Q$. However, it increases the size by a factor exponential in $r$.

We now show that, when $\F_q$ is sufficiently large, this drawback can be mitigated using techniques from polynomial identity testing. In particular, using the current best known result of \cite{DBLP:journals/cc/GurjarT20}, the size increases by only a quasi-polynomial factor in $r$, rather than an exponential one. Moreover, this approach also yields the stronger lossless rank extractor properties.

We begin by defining a relevant class of polynomials.

\begin{definition}[Symbolic determinants with rank-one summands]
Let $r$ and $N$ be positive integers, let $x=\{x_1,\dots,x_N\}$ be a set of variables, and let $\F$ be a field. 
Define $\VP{r}{N}{\F}\subseteq \F[x]$ to be the class of polynomials $f\in \F[x]$ that can be written as
\[
f(x)=\det\!\left(\sum_{i=1}^N x_i A_i\right),
\]
where each $A_i\in \F^{r\times r}$ satisfies $\rank(A_i)\le 1$.
We call such polynomials \emph{$(r,N)$ symbolic determinants with rank-one summands}.  
\end{definition}

We also need the notion of \emph{$\delta$-hitting sets}.

\begin{definition}[$\delta$-hitting sets]
Let $\mathcal{C}\subseteq \F[x_1,\dots,x_N]$ be a class of polynomials, and let $\delta\geq 0$. A set $\mathcal{H}\subseteq \F^{N}$ is called a $\delta$-hitting set for $\mathcal{C}$ if for every nonzero polynomial $f\in \mathcal{C}$,
\[
\bigl|\{a\in \mathcal{H} : f(a)\neq 0\}\bigr|\ge (1-\delta)|\mathcal{H}|.
\]
\end{definition}

The following lemma shows that an explicit lossless rank extractor over $\F_Q$ can be transformed into one over $\F_q$, assuming the existence of an explicit $\delta$-hitting set $\mathcal{H}\subseteq \F_q^{rd}$ for the class $\VP{r}{rd}{\F_q}$, where $d=[\F_Q:\F_q]$.

\begin{lemma}\label{lem:hittingsetVBPtoRankExtractor}
Let $r\le k$ be positive integers, and let $\F_Q/\F_q$ be a finite field extension
of degree $d=[\F_Q:\F_q]$. Let
$\mathcal{E}_Q\subseteq \F_Q^{r\times k}$ be an explicit $(k,r,L)$ strong (resp. weak) lossless
rank extractor over $\F_Q$ of size $n$, and let
$\mathcal{H}\subseteq \F_q^{dr}$ be an explicit $\delta$-hitting set for
$\VP{r}{dr}{\F_q}$. Then one can construct an explicit
$(k,r,|\mathcal{H}|(L+\delta(rn-L)))$ strong lossless rank extractor (resp. an explicit
$(k,r,|\mathcal{H}|(L+\delta(n-L)))$ weak lossless rank extractor) over
$\F_q$ of size $n|\mathcal{H}|$.
\end{lemma}

\begin{proof}
Fix a basis $e_1,\dots,e_d$ of $\F_Q$ over $\F_q$.
For $a\in\F_Q$, let $a^{(1)},\dots,a^{(d)}\in\F_q$ denote its coordinates
with respect to this basis.

For $E=(E_{ij})\in\mathcal{E}_Q$ and $i\in[r],j\in[d]$, let
$A_{i,j}(E)\in\F_q^{r\times k}$ be the matrix whose $i$-th row is
$(E_{i1}^{(j)},\dots,E_{ik}^{(j)})$ and whose other rows are zero.
Then $\rank(A_{i,j}(E))\le 1$. 
Let $x=\{x_{i,j} : i\in [r],\, j\in [d]\}$ be variables, and define
\[
E^*(x):=\sum_{i\in[r],\,j\in[d]}x_{i,j}A_{i,j}(E)
\in\F_q[x]^{r\times k},
\]
Define the collection of matrices
\[
\mathcal{E}_q := \{E^*(a) : E\in \mathcal{E}_Q,\ a\in \mathcal{H}\}\subseteq \F_q^{r\times k}.
\]
Then $|\mathcal{E}_q|=n|\mathcal{H}|$.

Fix a full-rank matrix $M\in\F_q^{k\times r}$. We first make the following
observation. Suppose $\rank(EM)=t$. Choose a nonsingular $t\times t$ submatrix
of $EM$. Let $B(x)\in\F_q[x]$ be the corresponding submatrix of $E^*(x)M$, using the
row set $I\subseteq [r]$.
Evaluating $E^*(x)$ at $x_{i,j}=e_j$ for all $i\in [r]$ and $j\in [d]$ yields $E$. So $\det(B(x))\neq 0$.

Form the $r\times r$
block-diagonal symbolic matrix
\[
C(x):=
\begin{pmatrix}
B(x) & 0\\
0 & \operatorname{diag}(x_{i,1}:i\notin I)
\end{pmatrix}.
\]
Then $\det(C(x))=\det(B(x))\prod_{i\notin I}x_{i,1}\neq 0$.
Moreover, $C(x)$ is a symbolic matrix with rank-one summands. Indeed, write $B(x)=\sum_{i\in I,\,j\in[d]}x_{i,j}B_{i,j}$,
where $B_{i,j}$ is the corresponding $t\times t$ submatrix of
$A_{i,j}(E)M$. Since $\rank(A_{i,j}(E)M)\le 1$, we have
$\rank(B_{i,j})\le 1$. After embedding each $B_{i,j}$ into the upper-left
$t\times t$ block of an $r\times r$ matrix, and using the rank-one matrix
having a single nonzero diagonal entry for each variable $x_{i,1}$ with
$i\notin I$, we may write $C(x)=\sum_{i\in[r],\,j\in[d]}x_{i,j}C_{i,j}$
with $\rank(C_{i,j})\le 1$ for all $i,j$. Hence $\det(C(x))\in\VP{r}{dr}{\F_q}$.

Therefore, by the $\delta$-hitting set property of $\mathcal{H}$, we have
\begin{equation}\label{eq:bada}
\bigl|\{a\in\mathcal{H}:\rank(E^*(a)M)<t\}\bigr|
\le \delta|\mathcal{H}|, \quad\text{where } t=\rank(EM).
\end{equation}

Suppose first that $\mathcal{E}_Q$ is an explicit $(k,r,L)$ weak lossless rank extractor. For at least $n-L$ matrices
$E\in\mathcal{E}_Q$, we have $\rank(EM)=r$. Applying \eqref{eq:bada} with $t=r$
shows that at most $|\mathcal{H}|L+\delta|\mathcal{H}|(n-L)
=|\mathcal{H}|(L+\delta(n-L))$
matrices in $\mathcal{E}_q$ have rank less than $r$ after multiplication
by $M$. So $\mathcal{E}_q$ is an explicit $(k,r,|\mathcal{H}|(L+\delta(n-L)))$ weak lossless rank extractor.

Now suppose that $\mathcal{E}_Q$ is an explicit $(k,r,L)$ strong lossless rank extractor. For $E\in\mathcal{E}_Q$, let $d_E:=r-\rank(EM)$.
Then $\sum_{E\in \mathcal{E}_Q} d_E\le L$. 

By \eqref{eq:bada}, for all but at most $\delta|\mathcal{H}|$ points
$a\in\mathcal{H}$, we have $r-\rank(E^*(a)M)\leq d_E$.
For the remaining points, we use the trivial bound
$r-\rank(E^*(a)M)\le r$. Therefore,
\[
\sum_{a\in\mathcal{H}}
\bigl(r-\rank(E^*(a)M)\bigr)
\le
(1-\delta)|\mathcal{H}|d_E+\delta|\mathcal{H}|r.
\]
Summing over $E\in\mathcal{E}_Q$ gives
\begin{align*}
\sum_{E\in\mathcal{E}_Q}\sum_{a\in\mathcal{H}}
\bigl(r-\rank(E^*(a)M)\bigr)
\le
\sum_{E\in\mathcal{E}_Q} ((1-\delta)|\mathcal{H}|d_E+\delta|\mathcal{H}|r)
&\leq (1-\delta)|\mathcal{H}| L+\delta|H|rn\\
&=
|\mathcal{H}|((L+\delta(rn-L)).
\end{align*}
So $\mathcal{E}_q$ is an explicit
$(k,r,|\mathcal{H}|(L+\delta(rn-L)))$ strong lossless rank extractor.
\end{proof}

Gurjar and Thierauf~\cite{DBLP:journals/cc/GurjarT20} constructed explicit hitting sets for symbolic determinants with rank-one summands. We adapt their construction to prove the following lemma. The proof is deferred to \cref{sec:hittingsetVBP1}.

\begin{lemma}[store=hittingset]\label{lem:alphaHittingset}
    Let  $N$ be a positive integer, $\delta\in (0,1)$, and $\F$ a field such that $|\F| \geq (N/\delta)^c$ for some large enough absolute constant $c>0$. Then, one can construct an explicit $\delta$-hitting set $\mathcal{H}\subseteq \F^N$ for $\VP{r}{N}{\F}$ of size polynomial in $(N/\delta)^{\log N}$.
\end{lemma}

By combining \cref{lem:hittingsetVBPtoRankExtractor} and
\cref{lem:alphaHittingset}, we obtain the following theorem, which is the
main result of this subsection.

\begin{theorem}\label{thm:Qtoq}
Let $\F_Q/\F_q$ be a finite field extension of degree $d=[\F_Q:\F_q]$.
Let $r\le k$ be positive integers, let $\delta\in(0,1)$, and let $q$ be
a prime power satisfying $q\ge (dr/\delta)^c$ for a sufficiently large
absolute constant $c>0$.
Let $\mathcal{E}_Q\subseteq\F_Q^{r\times k}$ be an explicit $(k,r,L)$
strong (resp.\ weak) lossless rank extractor over $\F_Q$ of size $n$.
Then there exists an explicit $(k,r,(L+\delta n)h)$ strong\footnote{In the case of strong lossless rank extractors, we apply
\cref{lem:alphaHittingset} with error parameter $\delta/r$ instead of
$\delta$ to absorb the extra factor of $r$ in \cref{lem:hittingsetVBPtoRankExtractor}.} (resp.\ weak) lossless
rank extractor $\mathcal{E}_q$ over $\F_q$ of size $nh$, where $h\le (dr/\delta)^{O(\log(dr))}$.
\end{theorem}

\subsection{Putting It Together} \label{sec:combine}

We begin by proving the following result over arbitrary finite fields via field reduction:

\begin{theorem}[Formal version of \cref{thm:rcsmall-informal}]\label{thm:rcsmall}
Let $1\le r< k$ be integers, and let $q>1$ be a prime power. Then there exists an explicit $(k,r)$ lossless rank disperser over $\F_q$ of size $n$, where
\[
n = O\!\left(\max\left\{\frac{c\log (2r)}{\log q},2\right\}^r r(k-r)q\right),
\]
for some absolute constant $c>0$. Moreover, for each fixed $r$ and $q$, we have
\[
n \le c_1 \max\left\{\frac{c\log (2r)}{\log q},2\right\}^r r(k-r)
\]
for some absolute constant $c_1>0$ and infinitely many $k$.
\end{theorem}

\begin{proof}
Choose the smallest integer $d\geq 2$ such that $Q=q^d$ satisfies $Q\ge (2r)^{c_0}$, where $c_0$ is the constant from \cref{thm:fullrankbound}. 
Apply \cref{thm:fullrankbound} to obtain an explicit $(k,r,L)$ strong lossless rank extractor $\mathcal{E}_0$ over $\F_Q$ of size $n_0 > L$.  If $d=2$, then $Q=q^2$ is a square, and the sharper estimate for square fields noted in the proof of \cref{thm:fullrankbound} gives $L+1=O(r(k-r)q)$.
If $d>2$, then by \cref{thm:fullrankbound} and the minimality of $d$, we have $L+1=O(r(k-r)Q)
=O\bigl(r(k-r)q(2r)^{c_0}\bigr)$.
Moreover, for each fixed $r$ and $q$, we have $L+1 \le c_1 r(k-r)$ for some absolute constant $c_1>0$ and infinitely many $k$.

Remove $n_0-(L+1)$ elements from $\mathcal{E}_0$. The resulting collection $\mathcal{E}$ is still a $(k,r)$ lossless rank disperser over $\F_Q$, and its size is $L+1$.

By \cref{thm:ReverseExtension}, the collection $\mathcal{E}_{\F_Q\to \F_q}$ defined in \cref{sec:det} is a $(k,r)$ lossless rank disperser over $\F_q$ of size
\[
d^r (L+1) \le \max\left\{\frac{c'\log (2r)}{\log q},2\right\}^r (L+1)=O\!\left(\max\left\{\frac{c\log (2r)}{\log q},2\right\}^r r(k-r)q\right)
\]
for some absolute constants $c,c'>0$. The bound for infinitely many $k$ follows similarly from $L+1=O(r(k-r))$.
\end{proof}

\begin{remark}\label{remark:two}
It follows from the proof above that the number $2$ in $\max\left\{\frac{c\log (2r)}{\log q},2\right\}$ can be improved to $1$ if $q$ is non-prime. The reason we choose $2$ in general is that we need $Q$ to be non-prime so that \cref{thm:fullrankbound} applies.
\end{remark}

Next, we prove the following result over prime fields:

\begin{theorem}[Formal version of \cref{thm:rankprime-informal}]\label{thm:rankprime}
Let $1\le r<k$ be integers, and let $\delta\in (0,1)$.
Suppose $q>1$ is a prime satisfying $q\ge (2r/\delta)^{c_0}$ for a sufficiently large absolute constant $c_0$. Then there exists an explicit $(k,r,L)$ strong lossless rank extractor over $\F_q$ of size $n$, where 
\[
L\le (2r/\delta)^{c\log (2r)} r(k-r)q
\]
for some absolute constant $c>0$, and $L/n\le \delta$.
Moreover, for each fixed $r$ and $q$, we have 
\[
L\le (2r/\delta)^{c\log (2r)} r(k-r)
\]
for infinitely many $k$.
\end{theorem}

\begin{proof}
Let $Q=q^2$. Apply \cref{thm:fullrankbound} to obtain an explicit $(k,r,L_0)$ strong lossless rank extractor $\mathcal{E}_0$ over $\F_Q$ of size $n_0 > L_0$, where $L_0 = O(r(k-r)q)$ and $L_0/n_0\le Q^{-1/4}=q^{-1/2}$. Moreover, for each fixed $r$ and $q$, we have $L_0\le c_1 r(k-r)$ for some absolute constant $c_1>0$ and infinitely many $k$.

Let $\delta'=\delta/2$. Since $q\ge (2r/\delta)^{c_0}$, we may assume $q^{-1/2}\le \delta'$.
By removing elements from $\mathcal{E}_0$, we obtain an explicit $(k,r,L_0)$ strong lossless rank extractor $\mathcal{E}_Q$ over $\F_Q$ of size $n_1=\lceil L_0/\delta'\rceil$.

By \cref{thm:Qtoq}, there exists an explicit $(k,r,L)$ strong lossless rank extractor $\mathcal{E}$ over $\F_q$ of size $n$, where $L=(L_0+\delta' n_1)h$, $n=n_1 h$, and $h\le (2r/\delta')^{O(\log(2r))}$. Then
\[
L/n\le (L_0/n_1)+\delta'\le 2\delta'=\delta.
\]
Moreover, since $L_0=O(r(k-r)q)$,
\[
L=(L_0+\delta'n_1)h \le (2L_0+\delta')h = (2L_0+\delta/2)h \le (2r/\delta)^{c\log (2r)} r(k-r)q,
\]
where $c>0$ is a sufficiently large absolute constant.
Finally, since $L_0\le c_1 r(k-r)$ for infinitely many $k$, by choosing $c$ large enough, we obtain
\[
L \le (2r/\delta)^{c\log (2r)} r(k-r)
\]
for infinitely many $k$.
\end{proof}

Next, we prove our main results on explicit subspace designs.

\begin{theorem}[Formal version of \cref{thm:design-nonprime-informal,thm:design-prime-informal,thm:design-constant-informal}]\label{thm:designbound}

Let $1\le r<k$ be integers, let $r'=\min\{r,k-r\}$, and let $\delta\in(0,1)$.
Suppose $q>1$ is a prime power. In each of the cases below, there exist explicit collections
$\mathcal{V}$ and $\mathcal{V}^\perp$ of size $n$, consisting respectively of $r$-dimensional
and $(k-r)$-dimensional subspaces of $\F_q^k$, such that
$\mathcal{V}=\{V^\perp: V\in\mathcal{V}^\perp\}$ and $\mathcal{V}^\perp=\{V^\perp: V\in\mathcal{V}\}$.
Moreover, the following hold.

\begin{enumerate}
\item\label{item:nonprime}
Suppose $q$ is a non-prime prime power satisfying
$q\ge (2r')^{c_0}$ for some sufficiently large absolute constant $c_0>0$.
Then $\mathcal{V}$ is a $(k-r,A)$ strong subspace design and
$\mathcal{V}^\perp$ is an $(r,A)$ strong subspace design, where
$A=O(r(k-r)q)$ and $n\ge q^{1/4}A$.
Moreover, for each fixed $r$ and $q$, we have
$A\le c_1r(k-r)$ for some absolute constant $c_1>0$ and infinitely many $k$.

\item\label{item:prime}
Suppose $q$ is a prime satisfying
$q\ge (2r'/\delta)^{c_0}$ for some sufficiently large absolute constant $c_0>0$.
Then $\mathcal{V}$ is a $(k-r,A)$ strong subspace design and
$\mathcal{V}^\perp$ is an $(r,A)$ strong subspace design, where
$A\le (2r'/\delta)^{c\log(2r')}r(k-r)q$
for some absolute constant $c>0$, and $n\ge A/\delta$.
Moreover, for each fixed $r$ and $q$, we have $A\le (2r'/\delta)^{c\log(2r')}r(k-r)$
for infinitely many $k$.

\item\label{item:small}
In general, $\mathcal{V}$ is a $(k-r,A)$ weak subspace design and
$\mathcal{V}^\perp$ is an $(r,A)$ weak subspace design, where
\[
A=n-1=
O\left(
\max\left\{\frac{c\log(2r')}{\log q},2\right\}^{r'}
r(k-r)q
\right)
\]
for some absolute constant $c>0$.
Moreover, for each fixed $r$ and $q$, we have
\[
A=n-1\le
c_1
\max\left\{\frac{c\log(2r')}{\log q},2\right\}^{r'}
r(k-r)
\]
for some absolute constant $c_1>0$ and infinitely many $k$.
\end{enumerate}
\end{theorem}

\begin{proof}
Note that if the theorem holds for $r\le k/2$, then the case $r>k/2$ follows by replacing $r$ with $k-r$ and swapping $\mathcal{V}$ and $\mathcal{V}^\perp$. Thus, we may assume $r\le k/2$, and hence $r'=r$. 

To prove \cref{item:nonprime}, we apply \cref{thm:fullrankbound} to obtain an explicit $(k,r,L)$ strong lossless rank extractor $\mathcal{E}_0$ consisting of full-rank matrices, with $L/|\mathcal{E}_0|\le q^{-1/4}$. 
Applying \cref{lem:rc2design} directly gives strong subspace designs
$\mathcal{V}$ and $\mathcal{V}^\perp$ of size $n=|\mathcal{E}_0|\ge q^{1/4}L$.
Taking $A=L$ proves \cref{item:nonprime}.

The proof of \cref{item:prime} is similar. We apply \cref{thm:rankprime} with error parameter $\delta/2$ to obtain an explicit $(k,r,L)$ strong lossless rank extractor $\mathcal{E}_0$ satisfying $L/|\mathcal{E}_0|\le \delta/2$, and then apply \cref{lem:reduced}. The resulting strong subspace designs have size
$n\ge |\mathcal{E}_0|-L\ge \frac{2L}{\delta}-L\ge \frac{L}{\delta}$.
Letting $A=L$, and increasing the absolute constants $c$ and $c_0$ if necessary to absorb the change from $\delta$ to $\delta/2$, proves \cref{item:prime}.

Finally, to prove \cref{item:small}, we apply \cref{thm:rcsmall} to obtain an explicit $(k,r)$ lossless rank disperser $\mathcal{E}_0$, and then remove all matrices of rank less than $r$. This does not affect the lossless rank disperser property: if $E\in\F_q^{r\times k}$ satisfies $\rank(EM)=r$ for some $M\in\F_q^{k\times r}$, then necessarily $\rank(E)=r$, and hence such matrices are not removed. Let $\mathcal{E}_1$ denote the resulting collection. We then apply \cref{lem:rc2design} to obtain $\mathcal{V}=\mathcal{V}(\mathcal{E}_1)$ and $\mathcal{V}^\perp=\mathcal{V}^\perp(\mathcal{E}_1)$. Let $A=|\mathcal{V}|-1=|\mathcal{V}^\perp|-1$, and \cref{item:small} follows.
\end{proof}

\begin{remark}
Using the same idea of exchanging $r$ with $k-r$ via duality, one may relax the condition $q\geq \poly(r)$ in \cref{thm:rankboundinformal} and \cref{thm:rankprime-informal} to $q\geq \poly(r')$, where $r'=\min\{r,k-r\}$. We choose not to adopt this formulation for simplicity.
\end{remark}

\begin{remark}[Comparison with previous subspace-design constructions] 
In our subspace-design constructions, the dimension of the test subspace equals the codimension of each member of the family. 
This property is crucial for the application to strong $s$-blocking sets (see \cref{thm:lre-to-sbs}).
The strong subspace designs of Guruswami, Xing, and Yuan~\cite{guruswami2018subspace} also cover this case. Over a fixed field $\F_q$ and for arbitrarily large ambient dimension $k$, however, their intersection bound contains an additional factor of $\log_q k$. Combined with Lemma~\ref{thm:lre-to-sbs}, their construction therefore gives strong $s$-blocking sets whose size is larger than the bounds in Corollary~\ref{cor:blocking} by this logarithmic factor. For ambient dimensions of the form $k=q^C-1$, where $C$ is fixed, this factor is absorbed into the constant. The subspace-design codes of Goyal, Guruswami, and Hsieh~\cite{GGH26}, on the other hand, give coordinate kernels whose codimension is substantially larger than the dimension of the test subspaces covered by their guarantee, and hence do not directly provide the equal-dimension/codimension designs required by Lemma~\ref{thm:lre-to-sbs}.
\end{remark}

\section{Explicit Constructions of Strong $s$-Blocking Sets}\label{sec:explicit-sbb}

Recall that a set $B \subseteq \PG(k-1,q)$ is a \emph{strong $s$-blocking set} if for every codimension-$s$ projective subspace $\Sigma$ of $\PG(k-1,q)$, the intersection $B \cap \Sigma$ spans $\Sigma$.
In this section, we present two types of explicit constructions of strong $s$-blocking sets in $\PG(k-1,q)$. The first arises from lossless rank extractors over $\F_q$, while the second is obtained from $\eps$-biased sets in $\F_q^k$ with sufficiently small bias. 

For a subset $S\subseteq \F_q^k$, let ${\rm span}(S)$ denote its $\F_q$-linear span, and let $\widetilde{S}\subseteq \PG(k-1,q)$ denote its projectivization, i.e., 
\[
\widetilde{S}=\{\langle u\rangle : u\in S\setminus\{0\}\},
\]
where $\langle u\rangle$ denotes the one-dimensional subspace spanned by $u$.

\subsection{Strong $s$-Blocking Sets via Lossless Rank Dispersers}

The following lemma shows that every $(k,s+1)$ lossless rank disperser over $\F_q$ (see \cref{defi:disperser}) gives rise to a strong $s$-blocking set in $\PG(k-1,q)$. This was essentially already observed in \cite[Proposition~10]{FZ16}.\footnote{A minor technical difference is that \cite[Proposition~10]{FZ16} assumes the matrices all have full rank, i.e., their row spaces have dimension $s+1$. This is unimportant, as rank-deficient matrices can be safely ignored.}

\begin{lemma}\label{thm:lre-to-sbs}
Let $s<k$ be positive integers.
Suppose that $\mathcal E=\{E_i\}_{i=1}^n\subseteq \F_q^{(s+1)\times k}$ is a $(k,s+1)$ lossless rank disperser over $\F_q$. For $i\in[n]$, let $V_i=\rowspan(E_i)\subseteq \F_q^k$. Then $B:=\bigcup_{i=1}^n \widetilde{V_i}$ is a strong $s$-blocking set in $\PG(k-1,q)$. Moreover, $|B|\le n\frac{q^{s+1}-1}{q-1}\le 2nq^s$.
\end{lemma}

\begin{proof}
Suppose for contradiction that $B$ is not a strong $s$-blocking set. Then there exists a codimension-$s$ subspace $L\subseteq \F_q^k$ and a codimension-$(s+1)$ subspace $H\subseteq L$ such that for all $i\in [n]$,
\begin{equation}\label{eq:containment}
V_i \cap L \subseteq H.
\end{equation}

Choose $M\in \F_q^{k\times (s+1)}$ of rank $s+1$ such that $\ker(M^T)=H$. Since $\mathcal E$ is a $(k,s+1)$ lossless rank disperser, there exists $E_i\in \mathcal E$ such that $\rank(E_iM)=s+1$. In particular, $\dim V_i=s+1$.

Since $\dim(V_i)+\dim(L)=k+1$, we have $V_i\cap L\neq\{0\}$. By \eqref{eq:containment}, this implies $V_i\cap H\neq\{0\}$.

On the other hand, $\rank(E_iM)=s+1$ implies that $M^T$ is injective on $V_i$, hence
\[
V_i \cap \ker(M^T)=\{0\},
\]
i.e., $V_i \cap H=\{0\}$, a contradiction.

Thus $B$ is a strong $s$-blocking set. The size bound is immediate.
\end{proof}

Combining \cref{thm:lre-to-sbs} with \cref{thm:fullrankbound}, \cref{thm:rankprime}, and \cref{thm:rcsmall}, respectively, we obtain the following results on explicit strong $s$-blocking sets.

\begin{corollary}[Formal version of \cref{thm:blocking-nonprime-informal,thm:blocking-prime-informal,thm:blocking-constant-informal}]\label{cor:blocking}
Let $1\le s<k$ be integers, and let $q>1$ be a prime power.
Then the following hold:
\begin{enumerate}
\item\label{item:blocking1} Suppose $q$ is a non-prime prime power satisfying $q\ge (2s)^{c_0}$ for some sufficiently large absolute constant $c_0>0$. Then there exists an explicit strong $s$-blocking set $B\subseteq \PG(k-1,q)$ of size $|B|=O(s(k-s) q^{s+1})$. Moreover, for each fixed $s$ and $q$, we have $|B|=O(s(k-s) q^s)$ for infinitely many $k$.

\item\label{item:blocking2} Suppose $q$ is prime and satisfies $q\ge (2s)^{c_0}$ for some sufficiently large absolute constant $c_0>0$. Then there exists an explicit strong $s$-blocking set $B\subseteq \PG(k-1,q)$ of size at most $(2s)^{c\log(2s)} (k-s) q^{s+1}$ for some absolute constant $c>0$. Moreover, for each fixed $s$ and $q$, we have $|B|\le (2s)^{c\log(2s)} (k-s) q^s$ for infinitely many $k$.

\item\label{item:blocking3} In general, there exists an explicit strong $s$-blocking set $B\subseteq \PG(k-1,q)$ of size
\[
O\!\left(\max\left\{\frac{c\log (2s)}{\log q},2\right\}^{s+1} s(k-s) q^{s+1}\right)
\]
for some absolute constant $c>0$. Moreover, for each fixed $s$ and $q$, we have
\[
|B|\leq c_1 \max\left\{\frac{c\log (2s)}{\log q},2\right\}^{s+1} s(k-s) q^s
\]
for some absolute constant $c_1>0$ and infinitely many $k$.
\end{enumerate}
\end{corollary}

\subsection{Strong $s$-Blocking Sets via $\eps$-Biased Sets}

In this subsection, we first show that every $\eps$-biased set $S\subseteq \mathbb{F}_q^k$ with sufficiently small bias forms a strong $s$-blocking set. We then obtain explicit strong $s$-blocking sets by applying known constructions of $\eps$-biased sets from \cite{Ta17,JM21}.

We begin with basic terminology of discrete Fourier analysis over $\mathbb{F}_q^n$; for more details, see, e.g., \cite{TV06,O14}. Let $G=\F_q^k$, and let $\widehat{G}$ be its character group consisting of the characters $\chi: G\to\C^\times$. The identity of $\widehat{G}$ is called the \emph{trivial} character, while the others are \emph{nontrivial} characters.

Let $L(G,\C)$ be the space of functions $f:G\to \C$.
Define the
inner product $\langle\cdot,\cdot\rangle$ on $L(G,\C)$ via $\langle f,g\rangle=\ex_{x\in G} \left[f(x)\overline{g(x)}\right]$, where $\overline{g(x)}$ denotes the complex conjugate of $g(x)$. The character group $\widehat{G}$ is an orthonormal basis of $L(G,\C)$. 

The orthogonality relation of characters states that for two characters $\chi,\psi\in\widehat{G}$, we have $\langle \chi,\psi\rangle=\ex_{x\in G} \left[\chi(x)\overline{\psi(x)}\right]={\bf 1}[\chi=\psi]$. This implies that $\ex_{x\in G} \left[\chi(x)\right]=0$ for every nontrivial character $\chi$.

For $f\in L(G,\C)$ and $\chi \in \widehat{G}$, define the Fourier coefficient $\widehat{f}(\chi)$ of $f$ to be
\[
\widehat{f}(\chi)=\langle f,\chi\rangle=\ex_{x\in G} \left[f(x)\overline{\chi(x)}\right].
\]
This defines a function $\widehat{f}: \widehat{G} \to \C$ sending $\chi$ to $\widehat{f}(\chi)$, called the \emph{Fourier transform} of $f$. Its $L_1$-norm is defined to be $\|\widehat{f}\|_1=\sum_{\chi\in \widehat{G}} |\widehat{f}(\chi)|$. 

As $\widehat{G}$ is an orthonormal basis of $L(G,\C)$, any $f\in L(G,\C)$ can be expanded as 
$f=\sum_{\chi\in G} \widehat{f}(\chi) \chi$.

\begin{definition}[$\eps$-biased set]\label{def:eps-biased-set}
Let $\eps>0$. A nonempty set $S\subseteq G$ is called an \emph{$\eps$-biased set} if 
\[
\left|\ex_{x\in S}[\chi(x)]\right|\leq \eps
\]
for all nontrivial characters $\chi\in \widehat{G}$.
\end{definition}

The following lemma is well-known.
\begin{lemma}\label{lem:epsfool}
Let $f: G\to \C$ be a function, and let $S\subseteq G$ be an $\eps$-biased set. Then
\[
\left|\ex_{x\in S}[f(x)] - \ex_{x\in G}[f(x)]\right|\leq \|\widehat{f}\|_1 \cdot \eps.
\]
\end{lemma}
\begin{proof}
Expand $f$ as $f=\sum_{\chi\in \widehat{G}} \widehat{f}(\chi)\,\chi$. The claim holds for each character $\chi\in \widehat{G}$ by the definition of $\eps$-biased sets. The result then follows by linearity of expectation.
\end{proof}

For a set $S\subseteq G$, let $1_S: G\to \C$ denote its indicator function, i.e., $1_S(x)=1$ if $x\in S$ and $1_S(x)=0$ otherwise. We will use the following lemma.

\begin{lemma}\label{lem:l1bound}
    Let $V\subseteq G$ be an affine subspace over $\F_q$. Then $\|\widehat{1_V}\|_1=1$.
\end{lemma}

\begin{proof}
We have $V=u+L$ for some vector $u\in G$ and some linear subspace $L\subseteq G$.
Let $H\subseteq \widehat{G}$ be the subgroup of characters vanishing on $L$, i.e., $H=\{\chi\in\widehat{G}: \chi|_L=1\}$. Note that $H$ is the kernel of the surjective map $\widehat{G}\to \widehat{L}$ that dualizes the inclusion $L\subseteq G$. So $|H|=|\widehat{G}|/|\widehat{L}|=|G|/|L|$.

For $\chi\in \widehat{G}$, we have
\begin{align*}
\widehat{1_V}(\chi)
&=\ex_{x\in G}\left[1_V(x)\overline{\chi(x)}\right] 
=\frac{1}{|G|}\sum_{x\in V}
\overline{\chi(x)}
=\frac{1}{|G|}\sum_{y\in L}
\overline{\chi(u+y)}\\
&=\frac{1}{|G|}\overline{\chi(u)}\sum_{y\in L}
\overline{\chi(y)}
=\frac{|L|}{|G|}\overline{\chi(u)}\ex_{y\in L}
\left[\overline{\chi(y)}\right].
\end{align*}
Since $\chi|_L$ is either the trivial character or nontrivial, we have
\[
\ex_{y\in L}
\left[\overline{\chi(y)}\right]=\begin{cases}
1, & \text{if $\chi|_L=1$, or equivalently, $\chi\in H$,}\\
0, & \text{otherwise}.
\end{cases}
\]
It follows that 
\[
\|\widehat{1_V}\|_1
=\sum_{\chi\in \widehat{G}} |\widehat{1_V}(\chi)|
=\sum_{\chi\in H} |\widehat{1_V}(\chi)|
=\frac{|H| |L|}{|G|}|\overline{\chi(u)}|
=\frac{|H| |L|}{|G|}=1. \qedhere
\]
\end{proof}

Recall that $B \subseteq \F_q^k$ is an affine $s$-blocking set if it intersects every affine subspace of codimension $s$.
The next theorem connects $\eps$-biased sets with affine and strong $s$-blocking sets.

\begin{theorem}\label{thm:biased-to-blocking} 
    Let $B\subseteq G$ be an $\eps$-biased set. Then the following statements hold:
    \begin{enumerate}
        \item\label{item:biased-1} If $\eps<q^{-s}$, then $B$ is an affine $s$-blocking set.
        \item\label{item:biased-2} If $\eps<\frac{q-1}{2q^{s+1}}$, then $\widetilde{B}$ is a strong $s$-blocking set.
    \end{enumerate}
\end{theorem}

\begin{proof}
For \cref{item:biased-1}, assume $\eps<q^{-s}$. Let $L$ be an affine subspace of codimension $s$. By \cref{lem:epsfool,lem:l1bound},
\[
\left|\frac{|B\cap L|}{|B|}-q^{-s}\right|
=
\left|\ex_{x\in B}[1_L(x)]-\ex_{x\in G}[1_L(x)]\right|
\le \eps\|\widehat{1_L}\|_1
=\eps.
\]
Hence
\[
\frac{|B\cap L|}{|B|}\ge q^{-s}-\eps>0,
\]
so $B\cap L\neq\emptyset$. Since this holds for every affine subspace $L$ of codimension $s$, it follows that $B$ is an affine $s$-blocking set.

For \cref{item:biased-2}, assume $\eps<\frac{q-1}{2q^{s+1}}$. It suffices to show that for every codimension-$s$ linear subspace $L\subseteq \F_q^k$ and every hyperplane $H$ of $L$, we have
\[
\widetilde{B}\cap(\widetilde{L}\setminus \widetilde{H})\neq\emptyset.
\]
Writing $1_{L\setminus H}=1_L-1_H$, \cref{lem:l1bound} and the triangle inequality give
\[
\|\widehat{1_{L\setminus H}}\|_1\le \|\widehat{1_L}\|_1+\|\widehat{1_H}\|_1=2.
\]
Applying \cref{lem:epsfool} with $f=1_{L\setminus H}$, we obtain
\[
\left|\frac{|B\cap (L\setminus H)|}{|B|} - \frac{q^{k-s}-q^{k-s-1}}{q^k}\right|
=
\left|\ex_{x\in B}[1_{L\setminus H}(x)]-\ex_{x\in G}[1_{L\setminus H}(x)]\right|
\le 2\eps.
\]
Therefore,
\[
\frac{|B\cap (L\setminus H)|}{|B|}\ge q^{-s}-q^{-s-1}-2\eps>0,
\]
and hence $B\cap (L\setminus H)\neq\emptyset$. Since $0\in H$, this implies $(B\setminus\{0\})\cap (L\setminus H)\neq\emptyset$,
and therefore
$\widetilde{B}\cap(\widetilde{L}\setminus \widetilde{H})\neq\emptyset$.
This proves \cref{item:biased-2}.
\end{proof}

\cref{thm:biased-to-blocking} shows that explicit constructions of $\eps$-biased sets with sufficiently small bias yield explicit constructions of affine and strong $s$-blocking sets. We now briefly recall the explicit constructions of $\eps$-biased sets over $\mathbb{F}_q^k$ that are relevant for our purposes. For $q=2$, a breakthrough result of Ta-Shma~\cite{Ta17} shows that there exist explicit $\eps$-biased sets $S\subseteq\mathbb{F}_2^k$ with
$|S|=O\left(\frac{k}{\eps^{2+o(1)}}\right)$.
Jalan and Moshkovitz~\cite{JM21} extended Ta-Shma's construction to every prime power $q$ as follows.

\begin{theorem}[{\cite{JM21}}]\label{thm:explicit-eps}
There exists an explicit $\eps$-biased set $S\subseteq\mathbb{F}_q^k$ with
\[
|S|\le \frac{C_1 k (\log q)^{C_2}}{\eps^{2+\alpha}},
\]
where
\[
\alpha\le C_3\left(\frac{\log\log(1/\eps)}{\log(1/\eps)}\right)^{1/3}=o_{1/\eps}(1),
\]
and $C_1,C_2,C_3>0$ are absolute constants.
\end{theorem}

Combining \cref{thm:biased-to-blocking} and \cref{thm:explicit-eps}, we obtain the main result of this section.

\begin{proposition}[Formal version of \cref{thm:blocking-eps-informal}]\label{prop:explicit-biased-to-explicit-blocking}
For every pair of positive integers $s<k$ and every prime power $q>1$, there exist an explicit affine $s$-blocking set in $\F_q^k$ and an explicit strong $s$-blocking set in $\PG(k-1,q)$, both of size at most
\[
C_1\,k\,(\log q)^{C_2}\,q^{\,2s + C_3 s \left(\frac{\log(s\log q)}{s\log q}\right)^{1/3}},
\]
where $C_1,C_2,C_3>0$ are absolute constants.
\end{proposition}

\begin{proof}
Choose $\eps=\Theta(q^{-s})$ sufficiently small so that $\eps<q^{-s}$ and $\eps<\frac{q-1}{2q^{s+1}}$. Let $B\subseteq\F_q^k$ be an explicit $\eps$-biased set of size at most $C_1\,k\,(\log q)^{C_2}\,q^{\,2s + C_3 s \left(\frac{\log(s\log q)}{s\log q}\right)^{1/3}}$,
whose existence is guaranteed by \cref{thm:explicit-eps}. Note that $|\widetilde{B}|\le |B|$. Applying \cref{thm:biased-to-blocking}, we conclude that $B$ is an affine $s$-blocking set and that $\widetilde{B}$ is a strong $s$-blocking set.
\end{proof}

\section{Concluding Remarks and Open Problems}

We conclude by highlighting several open problems that arise naturally from our work.

\begin{enumerate}
\item A natural open problem is to explicitly construct, for
$1\leq r\leq t\leq k$, codimension-$t$ subspaces of $\F_q^k$ forming a
strong $(r,A)$ subspace design with $A=O\left(\frac{rk}{t-r+1}\right)$
over a field whose size is bounded independently of $k$. The present paper treats the case $t=r$, where counting vanishing multiplicities gives strong designs with the same asymptotic parameters. Our current methods, however, do not exploit the extra codimension when $t>r$, and in particular do not obtain the improvement by a factor of $t-r+1$. Obtaining such a gap-dependent bound is crucial for applications such as the recent NC algorithm for bipartite matching~\cite{CGGRT26} and subspace-design codes \cite{chen2025explicit, Goyal2025OptimalPG, GGH26}.

For $k=q^d-1$, Guruswami, Xing, and Yuan~\cite{guruswami2018subspace} obtain $A=O\left(\frac{d r k}{t-r+1}\right)$,
which matches the target bound up to a constant when $d$ is fixed, but incurs an
additional factor of $\log_q k$ over a fixed field. Goyal, Guruswami, and Hsieh~\cite{GGH26} construct additive codes over an alphabet $\F_q^m$, where $m=\poly(r,1/\eps)q^{r^2}$,
whose coordinate kernels satisfy near-optimal strong intersection bounds for test subspaces of dimension at most $r$. Thus, their result applies when the codimension $t$ may be much larger than $r$, but does not address the range where $t$ is close to $r$.

\item The gap between the best-known upper and lower bounds on the sizes of strong and affine $s$-blocking sets remains roughly a factor of $O(s)$. Closing this gap would be of significant interest.

\item For every sufficiently large non-prime prime power $q$, we constructed explicit strong $s$-blocking sets of size $O(skq^s)$, matching the best-known existential bound up to an absolute constant factor. In contrast, when $q$ is prime or a small prime power, our explicit constructions have a worse dependence on $s$ in the leading coefficient. It would be very interesting to obtain improved explicit constructions in these cases as well.

\item In particular, when $q=O(1)$, do there exist explicit strong $s$-blocking sets of size $O(C(s) k q^s)$ with $C(s)$ subexponential in $s$? Note that \cref{thm:blocking-eps-informal} achieves $C(s)=\exp(O(s))$ in this regime.

\item When $q=O(1)$, do there exist explicit $(k,r,L)$ weak lossless rank extractors of size $n$ over $\F_q$ such that $1-\frac{L}{n}\geq 2^{-o(r)}$?
Applying the field reduction in \cref{thm:ReverseExtension} gives
$1-\frac{L}{n}\geq 2^{-O(r\log\log r)}$.
In contrast, \cref{thm:existence} shows that there exist non-explicit strong lossless rank extractors with
$1-\frac{L}{n}=\Omega(1)$ (and in fact $1-\frac{L}{n}=1-o_q(1)$).
\end{enumerate}

\section*{Acknowledgments} We thank Abhranil Chatterjee, Sumanta Ghosh, Rohit Gurjar, Ran Tao, Ben Lee Volk, and Chaoping Xing for helpful discussions. An early discussion of this project took place at the 2025 Oberwolfach workshop “New Mathematical Directions in Coding Theory,” and we thank the institute for its hospitality and the organizers for the opportunity to participate.

\bibliographystyle{alphaurl}
\bibliography{ref}

\appendix
\crefalias{section}{appendix}

\section{Existence of Non-Explicit Strong Lossless Rank Extractors}\label{sec:existence}

The main result of this section is the following theorem, establishing the existence of non-explicit lossless strong rank extractors with good parameters over any finite field $\F_q$.

\begin{theorem}\label{thm:existence}
There exists a function $\delta:\{2,3,4,\dots\}\to (0,1)$ such that
$\lim_{q\to\infty}\delta(q)=0$ and the following holds.
For every prime power $q$ and all integers $k>r>0$,
there exist $(k,r,L)$ strong lossless rank extractors over $\F_q$ of size $n$ such that
$L\leq (1+o_q(1))r(k-r)$ and $L/n\le \delta(q)$.
\end{theorem}

The proof is based on the probabilistic method. A similar result was obtained in \cite[Section~5.3]{For14}, but for a weaker notion of lossless rank extractors, which coincides with our notion of lossless rank dispersers (\cref{defi:disperser}).

We need the following lemmas.

\begin{lemma}\label{lem:ineq}
Let $q\geq 2$. Let $r>0$ be an integer. Then $\prod_{i=1}^r (1-q^{-i})\geq e^{-2/(q-1)}$.
\end{lemma}

\begin{proof}
For \(x\in[0,1/2]\), we claim that $\log(1-x)\ge -2x$.
Indeed, let $f(x):=\log(1-x)+2x$. Then
$f'(x)=2-\frac{1}{1-x}=\frac{1-2x}{1-x}\ge 0$
for all \(x\in[0,1/2]\). Since \(f(0)=0\), it follows that \(f(x)\ge 0\) for \(x\in[0,1/2]\), proving the claim.

Applying the claim to $x=q^{-i}$ gives $\log(1-q^{-i})\ge -2q^{-i}$ for $i\geq 1$.
Summing over \(i\in [r]\), we obtain
\[
\sum_{i=1}^r \log(1-q^{-i})
\ge -2\sum_{i=1}^r q^{-i}
\ge -2\sum_{i=1}^\infty q^{-i}
= -2/(q-1).
\]
Exponentiating both sides proves the lemma.
\end{proof}

\begin{lemma}\label{lem:expbound}
For $q\geq 2$, define
\[
z_q=\begin{cases}
q & q>2\\
q^{1/2} & q=2
\end{cases}
\]
and
\[
\beta_q=\begin{cases}
\log_q 2 & q>2\\
2\log_{2}(2^{1/2}-e^{-2}(2^{1/2}-1-2^{-3/2})) & q=2
\end{cases}
\]
so that $z_q>1$ and $0<\beta_q<1$.
Then for a uniformly distributed random matrix $A\in\F_q^{r\times r}$ over a finite field $\F_q$, we have 
\[
\ex\left[z_q^{r-\rank(A)}\right]\leq z_q^{\beta_q}.
\]
\end{lemma}
\begin{proof}
We have
\begin{equation}\label{eq:kernelsize}
\ex\left[q^{r-\rank(A)}\right]
=\ex\left[|\ker(A)|\right]
=\sum_{x\in \F_q^r} \Pr[Ax=0]
=1+\sum_{0\neq x\in \F_q^r} \Pr[Ax=0]
=1+(q^r-1)q^{-r}=2-q^{-r}.
\end{equation}
When $q>2$, this gives
\[
\ex\left[z_q^{r-\rank(A)}\right]
=\ex\left[q^{r-\rank(A)}\right]=2-q^{-r}\leq 2=z_q^{\beta_q},
\]
as claimed.

Now consider the case where $q=2$. 
Let $p=\Pr[\rank(A)=r]$.
By \cref{lem:ineq}, we have
\begin{equation}\label{eq:rankdeficiency}
p=\frac{\prod_{j=0}^{r-1}(q^r-q^j)}{q^{r^2}}
=\prod_{i=1}^r(1-q^{-i})
\geq e^{-2/(q-1)}=e^{-2}.
\end{equation}

As the function $f(x)=x^{1/2}$ is concave, its curve is below the tangent line at any point $x>0$.
In particular, for $x>0$,
\[
x^{1/2}=f(x)\leq f(2)+f'(2)(x-2)=2^{1/2}+\frac{x-2}{2^{3/2}}.
\]
Substituting $x=2^{r-\rank(A)}$ yields
\[
z_q^{r-\rank(A)}\leq 2^{1/2}+\frac{2^{r-\rank(A)}-2}{2^{3/2}}.
\]
Therefore,
\begin{equation}\label{eq:exp1}
\begin{aligned}
\ex\left[z_q^{r-\rank(A)}\right]
&=p+(1-p)\ex\left[z_q^{r-\rank(A)}\mid \rank(A)<r\right]\\
&\leq p+(1-p)2^{1/2}+2^{-3/2}(1-p)\ex\left[2^{r-\rank(A)}-2\mid \rank(A)<r\right].
\end{aligned}
\end{equation}
By \eqref{eq:kernelsize}, we have
\[
2\geq \ex\left[q^{r-\rank(A)}\right]=p+(1-p)\ex\left[2^{r-\rank(A)}\mid \rank(A)<r\right],
\]
or equivalently,
\begin{equation}\label{eq:exp2}
(1-p)\ex\left[2^{r-\rank(A)}-2\mid \rank(A)<r\right]\leq p.
\end{equation}
Combining \eqref{eq:exp1} and \eqref{eq:exp2} shows that
\begin{align*}
\ex\left[z_q^{r-\rank(A)}\right]
\leq p+(1-p)2^{1/2}+2^{-3/2}p
&=2^{1/2}-p(2^{1/2}-1-2^{-3/2})\\
&\stackrel{\eqref{eq:rankdeficiency}}{\leq} 2^{1/2}-e^{-2}(2^{1/2}-1-2^{-3/2})\\
&= z_q^{\beta_q}.\qedhere
\end{align*}

\end{proof}

We now prove \cref{thm:existence}.

\begin{proof}[Proof of \cref{thm:existence}]
Let $z_q$ and $\beta_q$ be as in \cref{lem:expbound}.
Let $M\in \F_q^{k\times r}$ be a full-rank matrix.
Let $E_1,\dots,E_n$ be independent and uniformly distributed over $\F_q^{r\times k}$. Then $E_1M,\dots,E_n M$ are independent and
uniformly distributed over $\F_q^{r\times r}$. Therefore,
\begin{equation}\label{eq:random-E}
\begin{aligned}
\Pr\left[\sum_{i=1}^n(r-\rank(E_iM))>L\right]
&=\Pr\left[z_q^{\sum_{i=1}^n(r-\rank(E_iM))}\geq z_q^{L+1}\right]\\
&\leq z_q^{-(L+1)}\ex\left[z_q^{\sum_{i=1}^n(r-\rank(E_iM))}\right]\\
&=z_q^{-(L+1)}\prod_{i=1}^n \ex\left[z_q^{r-\rank(E_iM))}\right]\\
&\leq z_q^{\beta_q n-(L+1)},
\end{aligned}
\end{equation}
where the last inequality follows from \cref{lem:expbound}.

Let $S$ be the set of full-rank matrices in $\F_q^{k\times r}$.
The group $\GL(r,q)$ acts on $S$ by right multiplication, and each orbit has size
$\prod_{j=0}^{r-1}(q^r-q^j)$. Let $M_1,\dots,M_t\in S$ be a complete set of orbit representatives.
Then
\begin{equation}\label{eq:eqk}
t
=\frac{|S|}{\prod_{j=0}^{r-1}(q^r-q^j)}
=\prod_{j=0}^{r-1}\frac{q^k-q^j}{q^r-q^j}
\le q^{r(k-r)}\prod_{i=1}^r\frac{1}{1-q^{-i}}
\le q^{r(k-r)}e^{2/(q-1)}
\end{equation}
by \cref{lem:ineq}.

Assume that
\begin{equation}\label{eq:m-constraint}
q^{r(k-r)}e^{2/(q-1)}\cdot
z_q^{\beta_q n-(L+1)}<1.
\end{equation}
Then \eqref{eq:random-E}, \eqref{eq:eqk}, and a union bound over
$M_1,\dots,M_t$ imply that there exist matrices $E_1,\dots,E_n\in\F_q^{r\times k}$
such that $\sum_{i=1}^n(r-\rank(E_iM))\le L$ for every $M\in\{M_1,\dots,M_t\}$.
So $\{E_1,\dots,E_n\}$ is a $(k,r,L)$ strong lossless rank extractor.

It remains to choose $\delta(\cdot)$, $L$, and $n$ so that \eqref{eq:m-constraint} holds and the stated bounds are satisfied.
Let $\delta(q)=\beta_q^{1/2}$. Then $\lim_{q\to\infty} \delta(q)=\lim_{q\to\infty}(\log_q 2)^{1/2}=0$.

Taking logarithms, \eqref{eq:m-constraint} is equivalent to
\begin{equation}\label{eq:m-constraint-2}
r(k-r)\ln q+\frac{2}{q-1}
+(\beta_q n-(L+1))\ln z_q<0.
\end{equation}

Let $n=\lceil L/\delta(q)\rceil$, so that $L/n\leq \delta(q)$. Note that $n\leq L/\delta(q)+1=L\beta_q^{-1/2}+1$. So \eqref{eq:m-constraint-2} holds as long as
\begin{equation}\label{eq:m-constraint-3}
r(k-r)\ln q+\frac{2}{q-1}< (1-\beta_q^{1/2}) \ln z_q \cdot L
\end{equation}
holds.

For $q=2$, \eqref{eq:m-constraint-3} holds for sufficiently large $L=O(r(k-r))$.
For $q>2$, using the fact that $\beta_q=\log_q 2$ and $z_q=q$, we see that \eqref{eq:m-constraint-3} again holds for sufficiently large $L=(1+o_q(1))r(k-r)$.
\end{proof}

\section{$\delta$-Hitting Sets for Symbolic Determinants with Rank-One Summands}
\label{sec:hittingsetVBP1}

In this section, we prove \cref{lem:alphaHittingset}, which constructs explicit $\delta$-hitting sets for symbolic determinants with rank-one summands in $\F_q^N$ over a sufficiently large field $\F_q$. For convenience, we first restate the lemma.

\getkeytheorem{hittingset}

Let $N\geq 1$ be an integer, let $x=\{x_1,\dots,x_N\}$ be variables, and let $\F$ be a field. 
Consider symbolic matrices of the form
\[
A(x)=\sum_{i\in [N]} x_i A_i\in\F[x_1,\dots,x_N]^{r\times r},
\]
where $A_i\in \F^{r\times r}$ and $\rank(A_i)\leq 1$ for all $i\in [N]$. Since each $A_i$ has rank at most one, it can be written as $u_iv_i^T$ for some vectors $u_i,v_i\in \F^r$. Let $U,V\in \F^{r\times N}$ such that for $i\in [N]$, the $i$-th column of $U$ and that of $V$ are $u_i$ and $v_i$, respectively. Let $X$ be the $N\times N$ diagonal matrix whose $i$-th diagonal entry is the variable $x_i$ for $i\in [N]$. It can be easily shown that $A=UXV^T$. By the Cauchy--Binet formula,
\[
\det(A)=\det(UXV^T)=\sum_{S\subseteq [N],|S|=r} \det((UX)_S)\det(V_S)=\sum_{S\subseteq [N],|S|=r} \det(U_S)\det(V_S)\prod_{i\in S}x_i.
\]
For a subset $S\subseteq [N]$, let $x_S$ denote $\prod_{i\in S}x_i$. Note that the coefficient of $x_S$ in $\det(A)$ is nonzero if and only if $\det(U_S)$ and $\det(V_S)$ are nonzero. In other words, the coefficient of $x_S$ in $\det(A)$ is nonzero 
if and only if $S$ is a common base of the linear matroids represented by the matrices $U$ and $V$. We denote these two matroids by $\mathcal{M}^A_1$ and $\mathcal{M}^A_2$.
For $A$ as above, define the support
\[
\supp(A):=\{ S\subseteq [N] \;\mid \; \text{the coefficient of $x_S$ in $\det(A)$ is nonzero} \}.
\]
From the above discussion, $\supp(A)$ is the set of common bases of the linear matroids represented by $U$ and $V$.

For $u,v\in \Z^t$, write $u>v$ if $u$ is lexicographically larger than $v$.

We also need the notion of \emph{isolating weight assignments}.

\begin{definition}[Isolating weight assignments]
Let $t\le N$, and let $A=\sum_{i\in [N]} x_i A_i$ with $A_i\in\F^{r\times r}$ and $\rank(A_i) \leq 1$ for $i\in [N]$. 
Let $w=(w_1,\dots,w_N)\in (\Z_{\ge 0}^t)^N$. 
For $S\subseteq [N]$, define $w(S):=\sum_{i\in S} w_i \in \Z_{\geq 0}^t$.
We say that $w$ is \emph{isolating} for $A$ if there exists a unique set $S\in \supp(A)$ maximizing $w(S)$ lexicographically.
A set $\mathcal W\subseteq (\Z_{\ge 0}^t)^N$ is isolating for $A$ if it contains such a vector $w$.
\end{definition}

The importance of isolating weight assignments is captured by the following claim.

\begin{claim}\label{cl:isolationImpliesnon-zero}
Let $A(x)=\sum_{i\in [N]} x_i A_i$ with $\rank(A_i)\leq 1$ and $\det(A)\neq 0$. 
Let $w\in (\Z_{\ge 0}^t)^N$ be isolating for $A$. 
Let $y=\{y_1,\dots,y_t\}$ be variables, and define $A_w(y)\in\F[y_1,\dots,y_t]^{r\times r}$ by substituting
\[
x_i \mapsto \prod_{j\in [t]} y_j^{w_i[j]},
\]
where $w_i[j]$ denotes the $j$-th coordinate of $w_i$.
Then $\det(A_w(y))\neq 0$.
\end{claim}

\begin{proof}
Write $\det(A)=\sum_{S\in \supp(A)} c_S x_S$ with $c_S\neq 0$. 
Under the substitution, the monomial $x_S$ maps to $\prod_{j\in [t]} y_j^{w(S)[j]}$. 
Since $w$ is isolating for $A$, there is a unique subset $S\subseteq [N]$ maximizing $w(S)$ lexicographically, so the corresponding monomial $\prod_{j\in [t]} y_j^{w(S)[j]}$ is not canceled by the others.
\end{proof}

We follow the construction of Gurjar and Thierauf~\cite{DBLP:journals/cc/GurjarT20}, which yields isolating weight assignments $w$ with $t=1$, i.e.,  $\det(A_w(y))$ is univariate. However, the resulting weights can be quasi-polynomial in $N$, which would force us to work over a finite field $\F_q$ of size quasi-polynomial in $N$. 
To avoid this, we adapt their method and set $t=\Theta(\log N)$, which allows us to use polynomially bounded weights.

As mentioned earlier,  $\supp(A)$ is a set of common bases of two linear matroids. We take a slight detour and discuss the result of \cite{DBLP:journals/cc/GurjarT20} on the construction of weight isolation assignments for matroid intersection. 

\paragraph{Matroid intersection polytopes.}
Let $\mathcal{M}_1$ and $\mathcal{M}_2$ be two matroids of rank $r$ on common ground set $E$ with rank function $r_1$ and $r_2$. Then, the convex hull of characteristic vectors of common bases in $\mathbb{R}^{|E|}$, denoted by $P(\mathcal{M}_1,\mathcal{M}_2)$, is described by the following linear program:
\[
x_e\geq 0 \text{ for each } e\in E, \quad
x(S)=\sum_{e\in S}x_e\leq r_i(S) \text{ for each } S\subseteq E,\ i\in \{1,2\}, \quad
x(E)=r.
\]

The following lemma \cite[Lemma~3.5]{DBLP:journals/cc/GurjarT20} gives a characterization for the tight inequalities in the LP for faces of the matroid intersection polytope.
\begin{lemma}\label{lem:FaceMI}
Let $\mathcal{M}_1$ and $\mathcal{M}_2$ be two matroids on a common ground set $E$ with rank functions $r_1$ and $r_2$, respectively. Let $F$ be a face of $P(\mathcal{M}_1,\mathcal{M}_2)$. Then, there exist two partitions $\mathcal{P}_1,\mathcal{P}_2$ of $E$ such that for any $S_1\in \mathcal{P}_1$ and $S_2\in \mathcal{P}_2$, there exists $\nu_{S_1},\mu_{S_2}\in \mathbb{N}$ such that for each $x\in F$, $x(S_1)=\nu_{S_1}$ and $x(S_2)=\mu_{S_2}$.

Moreover,\begin{enumerate}
    \item If for some $T\subseteq E$, $x(T)=r_1(T)$ for all $x\in F$, or $x(T)=r_2(T)$ for all $x\in F$, then $T$ is the union of sets from $\mathcal{P}_1$, respectively $\mathcal{P}_2$.
    \item If for some $e\in E$, $x_e=0 $ for all $x\in F$, then  there is a $S_1\in \mathcal{P}_1$ and $S_2\in \mathcal{P}_2$ such that $S_1=S_2=\{e\}$ and $\nu_{S_1}=\mu_{S_2}=0$.
\end{enumerate}
\end{lemma}

\begin{definition}\label{def:cycle}
Let $F$ be a face of the polytope $P(\mathcal{M}_1,\mathcal{M}_2)$ with the partitions $\mathcal{P}_1,\mathcal{P}_2$  from \cref{lem:FaceMI}. A sequence
$C = (e_1,e_2,\dots,e_{2\ell})$ of distinct elements of $E$ is called a \emph{cycle}
with respect to face $F$, if consecutive pairs are alternately in a set
from $\mathcal{P}_1$ and a set from $\mathcal{P}_2$. That is, for $i \in [\ell]$,
\begin{align*}
      e_{2i-1},e_{2i} &\in S_{2i-1} \text{ for some } S_{2i-1} \in \mathcal{P}_1, \\
      e_{2i},e_{2i+1} &\in S_{2i} \text{ for some } S_{2i}\in \mathcal{P}_2.
\end{align*}
\end{definition}

We denote the set of all cycles with respect to a face $F$ as $\mathcal{C}_F$. For a vector $u\in \Z^E$, let $F$ be the face of $P(\mathcal{M}_1,\mathcal{M}_2)$ that maximizes $u^T\cdot y$ over all the points $y$ in the polytope. Then, we denote $\mathcal{C}_F$ by $\mathcal{C}_u$. For a bipartite graph $G$ with edge set $E$ and a vector $u\in \Z^E$, the \emph{circulation} of a cycle $C=(e_1,e_2,\dots,e_{2\ell})$ with respect to $u$, denoted by $u(C)$, is defined by
\[
u(C):=\left|\sum_{i=1}^{2\ell} (-1)^i u[e_i]\right|.
\]

The following lemma summarizes key properties of $\mathcal{C}_u$. Additionally, it establishes a sufficient condition for a weight assignment on the ground set to guarantee the existence of a unique maximum-weight common base. Such a weight assignment is said to be \emph{isolating} for the two matroids.

\begin{lemma}[\cite[Corollary~3.10 and Lemma~3.14]{DBLP:journals/cc/GurjarT20}] \label{lem:nocycleIsolating}
Let $\mathcal{M}_1$ and $\mathcal{M}_2$ be two matroids defined on a common ground set $E$ of size $N$, and let $u\in \mathbb{Z}^E$.
\begin{itemize}
    \item If $\mathcal{C}_{u}$ contains no cycles of length at most $c$, then the number of cycles in $\mathcal{C}_{u}$ of length at most $2c$ is at most $N^4$.
    \item If $\mathcal{C}_{u}$ is empty, then $u$ is isolating for $\mathcal{M}_1$ and $\mathcal{M}_2$.
\end{itemize}
\end{lemma}

The following lemma combines Claim~3.18 and Lemma~3.19 of~\cite{DBLP:journals/cc/GurjarT20}.

\begin{lemma}\label{lem:GT1}
Let $\mathcal{M}_1,\mathcal{M}_2$ be two matroids on a common ground set $E$ of size $N$. 
Let $t=\lceil \log_2 N\rceil$, and let $H$ be a positive integer. 
Let $W_0\in \Z^N$ be the all-ones vector.

For each $i\in [t]$, let $u_i\in \Z_{\geq 0}^N$ satisfy $H > N \max_{j\in [N]} u_i[j]$ and
\[
u_i(C)\neq 0 \quad \text{for every cycle $C$ in $\mathcal{C}_{W_{i-1}}$ of length at most $2^i$.}
\]
Define $W_i = \sum_{j=1}^i H^{\,i-j} u_j$.
Then $\mathcal{C}_{W_i}$ contains no cycles of length at most $2^i$ for every $i\in [t]$. Moreover, $W_t$ is an isolating weight assignment for the set of common bases of $\mathcal{M}_1$ and $\mathcal{M}_2$.
\end{lemma}

For a symbolic matrix $A$ with rank-one summands, if we apply the above lemma for matroids $\mathcal{M}^A_1$ and $\mathcal{M}^A_2,$ we get an $N$-tuple of vectors in one dimension that is isolating for $A$. The following corollary directly follows from the above lemma and gives an $N$-tuple of vectors in $t=\lceil \log_2 N\rceil$ dimension that is isolating for $A$.

\begin{corollary}\label{cor:GTA}
Let $A(x)=\sum_{i\in [N]} x_i A_i$ with $\rank(A_i)\le 1$ for all $i\in[N]$, and suppose $\det(A)\neq 0$. 
Let $t=\lceil \log_2 N\rceil$, and let $H$ be a positive integer. 
Let $W_0\in \Z^N$ be the all-ones vector. 
Let $\mathcal{M}^A_1$ and $\mathcal{M}^A_2$ be matroids such that $\supp(A)$ is the set of their common bases.

For each $i\in [t]$, let $u_i\in \Z_{\geq 0}^N$ satisfy $H > N \max_{j\in [N]} u_i[j]$ and
\[
u_i(C)\neq 0 \quad \text{for every cycle $C$ in $\mathcal{C}_{W_{i-1}}$ of length at most $2^i$.}
\]
Define $W_i = \sum_{j=1}^i H^{\,i-j} u_j$.
Then $\mathcal{C}_{W_i}$ contains no cycles of length at most $2^i$ for all $i\in[t]$.

Moreover, for each $i\in [N]$, define $w_i\in \Z_{\geq 0}^t$ by $w_i[j]=u_j[i]$ for all $j\in[t]$. Then $w=(w_1,\dots,w_N)\in (\Z_{\geq 0}^t)^N$ is isolating for $A$.
\end{corollary}

\begin{proof}
By \cref{lem:GT1}, the weight vector $W_t = \sum_{j=1}^t H^{\,t-j} u_j$ admits a unique maximum-weight common base for $\mathcal{M}^A_1$ and $\mathcal{M}^A_2$.

For a subset $S\subseteq [N]$, write $w(S) = \sum_{i\in S} w_i \in \Z_{\geq 0}^t$, and compare such vectors in lexicographic order. By definition,
\[
W_t(S) = \sum_{j=1}^t H^{\,t-j} \sum_{i\in S} u_j[i].
\]
Since $H > N \max_{i,j} u_j[i]$, there is no carry between the coefficients, and therefore for any two subsets $S_1,S_2\subseteq [N]$,
\[
W_t(S_1) > W_t(S_2)
\quad \Longleftrightarrow \quad
w(S_1) > w(S_2)
\]
in lexicographic order.

Thus, the unique maximum-weight common base under $W_t$ is also the unique maximum under $w$, and hence $w$ is isolating for $A$.
\end{proof}

\begin{lemma} \label{lem:isolatingw}
Let $N$ be a positive integer, $t=\lceil \log_2 N\rceil$, and $0<\epsilon<1$. 
Let $\mathcal{U}=\{u_1,u_2,\dots,u_D\}\subseteq \Z_{\ge 0}^N$ be a collection such that for any graph $G$ with at most $N$ edges and any set $\mathcal{C}$ of at most $N^4$ cycles of $G$,
\[
\Pr_{u\sim \mathcal{U}}[u(C)\neq 0 \ \forall C\in \mathcal{C}] \ge 1-\epsilon.
\]

Define $\mathcal{W}\subseteq (\Z_{\ge 0}^t)^N$ as follows: for each $T=(a_1,\dots,a_t)\in [D]^t$, let $w^T=(w^T_1,\dots,w^T_N)$ with
\[
w^T_i[j]=u_{a_j}[i] \quad \text{for all } i\in[N],\ j\in[t].
\]
Then $\mathcal{W}=\{w^T : T\in [D]^t\}$.

Then, for any $A(x)=\sum_{i\in[N]} x_i A_i$ with $\rank(A_i)\le 1$ and $\det(A)\neq 0$, at least a $(1-\epsilon)^t$ fraction of $w\in \mathcal{W}$ are isolating for $A$.
\end{lemma}
 
\begin{proof}
Let $T=(a_1,a_2,\dots,a_t)\sim [D]^t$ be chosen uniformly and independently. 
Let $H$ be an integer greater than $N\max_{i\in [t],\,j\in [N]} u_{a_i}[j]$. 
Let $W_0\in \Z^N$ be the all-ones vector, and for each $i\in [t]$, define $W_i=\sum_{j=1}^i H^{\,i-j} u_{a_j}$.

Let $\mathcal{M}^A_1$ and $\mathcal{M}^A_2$ be the matroids such that $\supp(A)$ is the set of their common bases. 
By assumption on $\mathcal{U}$, for any set $\mathcal{C}$ of cycles of size at most $N^4$,
\[
\Pr_{u\sim \mathcal{U}}[u(C)\neq 0 \ \forall C\in \mathcal{C}] \ge 1-\epsilon.
\]

From \cref{lem:nocycleIsolating}, if $\mathcal{C}_{W_{i-1}}$ has no cycles of length at most $2^{i-1}$, then it contains at most $N^4$ cycles of length at most $2^i$. Hence, for each $i\in [t]$,
\[
\Pr_{a_i\sim[D]}\!\left[u_{a_i}(C)\neq 0 \ \forall C\in \mathcal{C}_{W_{i-1}},\ |C|\le 2^i
\ \middle|\
\mathcal{C}_{W_{i-1}}\text{ has no cycles of length }\le 2^{i-1}\right]\ge 1-\epsilon.
\]

We now prove by induction that for every $i\in\{0,1,\dots,t\}$,
\[
\Pr_{(a_1,\dots,a_i)\sim[D]^i}\big[\mathcal{C}_{W_i}\text{ has no cycles of length }\le 2^i\big]\ge (1-\epsilon)^i.
\]

The base case $i=0$ is trivial. For the inductive step,
\begin{align*}
&\Pr_{(a_1,\dots,a_i)}\big[\mathcal{C}_{W_i}\text{ has no cycles of length }\le 2^i\big] \\
&\quad \ge 
\Pr\Big[
u_{a_i}(C)\neq 0 \ \forall C\in \mathcal{C}_{W_{i-1}},\ |C|\le 2^i 
\ \text{and}\ 
\mathcal{C}_{W_{i-1}}\text{ has no cycles of length }\le 2^{i-1}
\Big] \\
&\quad =
\Pr\Big[
u_{a_i}(C)\neq 0 \ \forall C\in \mathcal{C}_{W_{i-1}},\ |C|\le 2^i
\ |\
\mathcal{C}_{W_{i-1}}\text{ has no cycles of length }\le 2^{i-1}
\Big] \\
&\qquad \times
\Pr\big[\mathcal{C}_{W_{i-1}}\text{ has no cycles of length }\le 2^{i-1}\big] \\
&\quad \ge (1-\epsilon)\cdot (1-\epsilon)^{i-1} = (1-\epsilon)^i,
\end{align*}
where the first inequality holds by \cref{lem:GT1}, and the last inequality follows from the conditional bound above and the induction hypothesis.

Thus, with probability at least $(1-\epsilon)^t$, the graph $\mathcal{C}_{W_t}$ has no cycles of length at most $2^t\ge N$, and hence is empty. By \cref{lem:nocycleIsolating}, $W_t$ is isolating for $\mathcal{M}^A_1$ and $\mathcal{M}^A_2$. By \cref{cor:GTA}, this implies that $w^T$ is isolating for $A$.
\end{proof}

Now we construct the collection $\mathcal{U}$ that will be used together with \cref{lem:isolatingw} to obtain an isolating set.

\begin{claim}\label{cl:circulationWeights}
Let $N,s$ be positive integers, and let $0<\epsilon<1$. Then there exists a collection $\mathcal{U}\subseteq \Z_{\ge 0}^{N}$ such that for any bipartite graph $G$ with edge set $[N]$ and any set $\mathcal{C}$ of at most $s$ cycles in $G$, at least a $(1-\epsilon)$ fraction of the vectors $u\in \mathcal{U}$ satisfy $u(C)\neq 0$ for every $C\in \mathcal{C}$. Moreover,
\[
|\mathcal{U}|,\ \max_{u\in \mathcal{U},\, j\in [N]} u[j] \le \frac{2sN}{\epsilon}.
\]
\end{claim}

\begin{proof}
Let $p$ be a prime satisfying $\frac{sN}{\epsilon}\le p< \frac{2sN}{\epsilon}$, which exists by Bertrand's postulate. For each $a\in \{0,1,\dots,p-1\}$, define a vector $u_a\in \{0,1,\dots,p-1\}^N$ by
\[
u_a[i]=a^i \bmod p \qquad \text{for all } i\in [N].
\]
Let $\mathcal{U}=\{u_a:0\le a<p\}\subseteq \Z_{\ge 0}^N$. Then $|\mathcal{U}|=p<2sN/\epsilon$, and for every $u\in\mathcal{U}$ and $j\in[N]$, we have $u[j]<p<\frac{2sN}{\epsilon}$, as desired.

Let $C=(e_1,e_2,\dots,e_{2\ell})$ be a cycle in $G$, where each $e_i\in [N]$. Define
\[
f_C(y)=\sum_{i=1}^{\ell} y^{e_{2i-1}}-\sum_{i=1}^{\ell} y^{e_{2i}} \in \Z[y],
\]
which is nonzero since the edges of $C$ are distinct.

For any $a\in \{0,1,\dots,p-1\}$, we have
\[
f_C(a)\equiv \sum_{i=1}^{\ell} u_a[e_{2i-1}] - \sum_{i=1}^{\ell} u_a[e_{2i}] \pmod p.
\]
In particular, if $f_C(a)\not\equiv 0\pmod p$, then $u_a(C)\neq 0$.

Now let $f_{\mathcal{C}}=\prod_{C\in \mathcal{C}} f_C$. Since each $f_C$ is nonzero, we have $f_{\mathcal{C}}\neq 0$. Its degree is at most $sN$. Hence $f_{\mathcal{C}}$ has at most $sN$ roots in $\F_p$. Therefore, there are at least $p-sN$ values $a\in \{0,1,\dots,p-1\}$ such that $f_{\mathcal{C}}(a)\not\equiv 0\pmod p$. For each such $a$, we have $f_C(a)\not\equiv 0\pmod p$ for all $C\in\mathcal{C}$, and hence $u_a(C)\neq 0$ for all $C\in\mathcal{C}$. Since $p\ge sN/\epsilon$, we have $\frac{p-sN}{p}\ge 1-\epsilon$. Therefore, at least a $(1-\epsilon)$ fraction of the vectors in $\mathcal{U}$ satisfy $u(C)\neq 0$ for every $C\in\mathcal{C}$.
\end{proof}

Now we complete the proof of \cref{lem:alphaHittingset}.

\begin{proof}[Proof of \cref{lem:alphaHittingset}]
Let $t=\lceil \log_2 N\rceil$. Let $\mathcal{U}\subseteq \Z_{\geq 0}^N$ be given by \cref{cl:circulationWeights} with parameters $s=N^4$ and $\epsilon=\delta/(2t)$, and let $\mathcal{W}\subseteq (\Z_{\geq 0}^{t})^N$ be the corresponding collection from \cref{lem:isolatingw}. Let $S\subseteq \F$ be any set of size at least $\frac{8N^6t^2}{\delta^2}$. Such a set exists by the assumption on $|\F|$.

Define
\[
\mathcal{H}=\left\{ a_{w,v}\in \F^N \;\middle|\; w=(w_1,\dots,w_N)\in \mathcal{W},\ v\in S^t,\ a_{w,v}[i]=\prod_{j\in [t]} v_j^{\,w_i[j]} \right\}.
\]
We claim that $\mathcal{H}$ is a $\delta$-hitting set for $\VP{r}{N}{\F}$.

Let
\[
A(x)=\sum_{i\in [N]} x_i A_i
\]
with $A_i\in \F^{r\times r}$, $\rank(A_i)\le 1$ for all $i\in [N]$, and $\det(A)\neq 0$. By \cref{lem:isolatingw}, at least a $\left(1-\frac{\delta}{2t}\right)^t \ge 1-\frac{\delta}{2}$ fraction of the elements of $\mathcal{W}$ are isolating for $A$. Fix such a weight assignment $w=(w_1,\dots,w_N)$, where $w_i\in \Z_{\geq 0}^t$. Let $A_w(y)$ be the matrix obtained from $A$ by replacing each $x_i$ with $\prod_{j\in [t]} y_j^{\,w_i[j]}$. By \cref{cl:isolationImpliesnon-zero}, we have $\det(A_w(y))\neq 0$.

By the construction of $\mathcal{W}$ and the properties of $\mathcal{U}$ from \cref{cl:circulationWeights}, each entry of $A_w(y)$ is a polynomial in $y_1,\dots,y_t$ of total degree at most $t\max_{u\in \mathcal{U},\, j\in [N]} u[j]\le \frac{4N^5t^2}{\delta}$. Since $r\le N$, it follows that
\[
\deg(\det(A_w(y))) \le \frac{4N^6t^2}{\delta}.
\]

By construction of $\mathcal{H}$, for every $v\in S^t$ we have $\det(A)(a_{w,v})=\det(A_w(y))(v)$. Therefore, by the Schwartz--Zippel lemma, for at least a $1-\frac{\deg(\det(A_w(y)))}{|S|}\ge 1-\frac{\delta}{2}$ fraction of points $v\in S^t$, we have $\det(A_w(y))(v)\neq 0$.

Combining the two bounds, for at least a $\left(1-\frac{\delta}{2}\right)^2 \ge 1-\delta$ fraction of points in $\mathcal{H}$, the polynomial $\det(A)$ evaluates to a nonzero value. Hence $\mathcal{H}$ is a $\delta$-hitting set for $\VP{r}{N}{\F}$.

Finally, by construction, 
\[
|\mathcal{H}|=|\mathcal{W}|\cdot |S|^t = |\mathcal{U}|^t\cdot |S|^t
= \left(\frac{N}{\delta}\right)^{O(\log N)}. \qedhere
\]
\end{proof}

\end{document}